%% file: tocs.tex
\documentclass[envcountsame,smallextended]{svjour3}     
\smartqed  
\usepackage{graphicx}
\usepackage{amsmath}
\usepackage{amsfonts}
\usepackage{amssymb}
\usepackage{booktabs}
\usepackage{graphics}
\usepackage{color}
\usepackage{framed}
\usepackage{enumerate}
\usepackage[pdfborder={0 0 0}]{hyperref}
\usepackage{tabularx}
\usepackage{comment}
\usepackage{bbm}
\def\emm#1,{{\em #1}}

\newcommand{\beq}{\begin{equation}}
\newcommand{\eeq}{\end{equation}}

\begin{document}

\title{XML Compression via DAGs}

\author{Mireille Bousquet-M\'{e}lou \and
        Markus Lohrey         \and
        Sebastian Maneth      \and
        Eric Noeth
}

\institute{         
           M. Bousquet-M\'{e}lou \at
              CNRS, LaBRI, Universit\'{e} de Bordeaux\\
              Tel.: +33-5-40-00-69-06 \\
              \email{bousquet@labri.fr}           
           \and
	   M. Lohrey \at
              University of Leipzig \\
              Tel.: +49-341-97-32201\\
              \email{lohrey@informatik.uni-leipzig.de}           
           \and
           S. Maneth \at
              University of Edinburgh \\
              Tel.: +44-131-651-5642 \\
              \email{smaneth@inf.ed.ac.uk}              
           \and
           E. Noeth \at
              University of Leipzig \\
              Tel.: +49-341-97-32212\\
              \email{noeth@informatik.uni-leipzig.de}           
}

\date{Received: date / Accepted: date}

\maketitle

\begin{abstract}
Unranked trees can be represented using their minimal 
dag (directed acyclic graph). For XML this achieves high compression 
ratios due to their repetitive mark up. Unranked trees are often represented 
through first child/next sibling (fcns) encoded binary trees. 
We study the difference in size (= number of edges) 
of minimal dag versus minimal dag of the fcns encoded binary tree.
One main finding is that the size of the dag of the binary tree
can never be smaller than the square root of the size of the
minimal dag, and that there are examples that match this bound.
We introduce a new combined structure, the \emph{hybrid dag},
which is guaranteed to be smaller than (or equal in size to) both dags. 
Interestingly, we find through experiments 
that last child/previous sibling encodings are much better for 
XML compression via dags, than fcns encodings. 
We determine the average sizes of unranked and binary dags
over a given set of labels (under uniform distribution)
in terms of their exact generating functions, and in terms
of their asymptotical behavior.
\end{abstract}

\section{Introduction}
The tree structure of an XML document can be conveniently 
represented as an ordered unranked 
tree~\cite{DBLP:conf/dbpl/Suciu01,DBLP:journals/sigmod/Neven02}.
For tree structures of common XML documents
\emph{dags} (\emph{directed acyclic graphs}) offer high compression ratios:
the number of edges of the minimal dag is
only about 10\% of the number of edges of the original 
unranked tree~\cite{DBLP:conf/vldb/KochBG03}
In a minimal dag, each distinct subtree is represented only once.
A dag can be exponentially smaller than the represented tree.
Dags and their linear average time construction via hashing are 
folklore in computer science (see e.g.~\cite{DBLP:journals/cacm/Ershov58});
they are a popular data structure used for sharing of common subexpressions
(e.g., in programming languages) and in
binary decision diagrams, see~\cite{DBLP:books/sp/MeinelT98}.
Through a clever pointer data structure, worst-case linear time
construction is shown in~\cite{DBLP:journals/jacm/DowneyST80}.

Unranked trees of XML tree structures are often
represented using binary trees, see~\cite{DBLP:journals/jcss/Schwentick07}
for a discussion.
A common encoding is the 
\emph{first child/next sibling encoding}~\cite{DBLP:conf/vldb/Koch03}
(in fact, this encoding is well-known,
see Paragraph~2.3.2 in Knuth's first book~\cite{DBLP:books/aw/Knuth68}).
The binary tree $\text{fcns}(t)$ is  
obtained from an unranked tree $t$ as follows.
Each node of $t$ is a node of $\text{fcns}(t)$.
A node $u$ is a left child of node $v$ in $\text{fcns}(t)$ if and only
if $u$ is the first child of $v$ in $t$.
A node $u$ is the right child of a node $v$ in 
$\text{fcns}(t)$ if and only if $u$ is the next sibling
of $v$ in $t$. 
From now on, when we speak of the size of a graph
we mean its number of edges.
Consider the minimal dag of $\text{fcns}(t)$ (called \emph{bdag} for 
\emph{binary dag} in the following)
in comparison to the minimal dag of $t$.
It was observed in~\cite{DBLP:journals/is/BusattoLM08} that
the sizes of these dags may differ, in both directions.
For some trees the difference is dramatic, which motivates the work
of this paper: to study the precise relationship between the
two dags, and to devise a new data structure that is guaranteed to be
of equal or smaller size than the minimum size of the two dags. 

Intuitively, the dag of $t$ shares \emph{repeated subtrees},
while the dag of $\text{fcns}(t)$ shares 
\emph{repeated sibling end sequences}. 
\begin{figure}[ht]
\centerline{
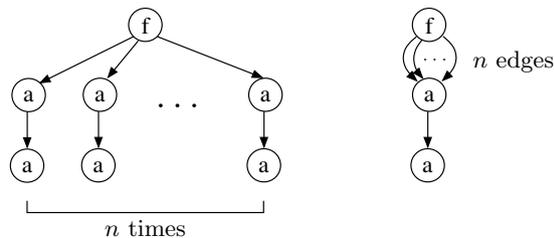
}
\caption{The unranked tree $t_n$ and $\text{dag}(t_n)$.}
\label{fig:tn}
\end{figure}
Consider the tree $t_n$ in the left of Figure~\ref{fig:tn}.
Its minimal dag is shown on the right.
As can be seen, each repeated subtree  is
removed in the dag. 
The dag consists of $n+1$ edges while $t_n$
consists of $2n$ edges. Moreover, $\text{fcns}(t_n)$ does not have any
repeated subtrees (except for leaves), i.e., the bdag of $t_n$ has $2n$ edges as well.
\begin{figure}[ht]
\centerline{
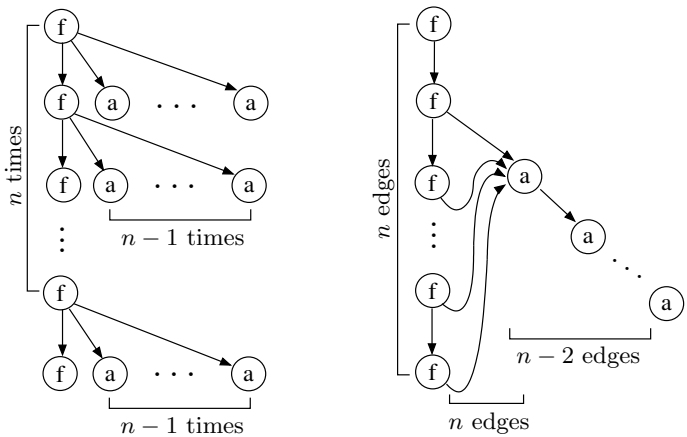
}
\caption{The unranked tree $s_n$ and $\text{bdag}(s_n)$.}
\label{fig:tn2}
\end{figure}
Next, consider the tree $s_n$ in the left of Figure~\ref{fig:tn2}.
Its bdag is shown on the right, it has $3n-2$ edges.
On the other hand, $s_n$ has $n^2$ edges and the same is 
true for the dag of $s_n$ since this tree has no repeated subtrees (except for leaves).
These two examples show that
(i)~the size of the dag of an unranked tree can be half the size
of the dag of the fcns encoded tree and
(ii)~the size of the dag of the fcns encoded tree can be
quadratically smaller than the size of the dag of the unranked tree. 
We prove in this paper that these ratios
are maximal: The size of the dag of the unranked tree is 
(i)~lower bounded by half of the size of the bdag 
and (ii)~upper bounded by the square of the size of the bdag.
Actually, we derive these bounds from stronger statements concerning
a combination of the unranked dag and the binary dag, called the
\emph{hybrid dag}, which combines both ways of sharing.
The idea is as follows. Given an unranked tree, we
compute its minimal dag.  The dag can be naturally viewed as
a regular tree grammar: Introduce for each node $v$ of the dag  
a nonterminal $A_v$ for the grammar. If a node $v$ is labeled
with the symbol $f$ and its children in the dag are $v_1,\ldots, v_n$ in this
order, then we introduce the production $A_v \to f(A_{v_1},\ldots,A_{v_n})$.
We now apply the fcns encoding to all
right-hand sides of this grammar.
Finally, we compute the minimal dag of the forest consisting
of all these fcns encoded right-hand sides.
See Figure~\ref{fig:hdag} which shows a tree $t$ of size $9$.
Its unranked and binary dags are each of size $6$.
The hybrid dag consists of a start tree plus one rule,
and is of size $5$.
For the XML document trees of our corpus, the average size of the
hybrid dag is only $76\%$ of the average size of the unranked dag.

We show that the size of the hybrid dag is always bounded by the minimum
of the sizes of the unranked dag and the binary dag. Moreover, 
we show that (i) the size of the hdag  is at least
half of the size of the binary dag and (ii) 
the size of the unranked dag is at most the square of the size of the hdag.
The above mentioned bounds for the unranked dag and binary dag are 
direct corollaries of these bounds.

The tree grammar of a hybrid dag is not a regular tree grammar 
anymore (because identifier nodes may have a right child).
It can be unfolded in three passes:
first undoing the sharing of tree sequences, 
then the binary decoding, and then undoing  sharing of
subtrees.
We show that these grammars can be translated into a well
known type of grammars: straight-line linear context-free tree 
grammars, for short \emph{SLT grammars} 
(produced by BPLEX~\cite{DBLP:journals/is/BusattoLM08} 
or TreeRePair~\cite{lohmanmen13}). 
This embedding increases the size only slightly.
One advantage is that SLT grammars can be
unfolded into the original tree in one pass.
Moreover, it is known that finite tree automata (even with
sibling equality constraints) and tree walking automata can
be executed in polynomial time over trees represented
by SLT grammars~\cite{DBLP:journals/tcs/LohreyM06,DBLP:journals/jcss/LohreyMS12,DBLP:journals/corr/abs-1012-5696}.

While in the theoretical limit the binary dag
can be smaller in comparison than the dag, it was observed 
in~\cite{DBLP:journals/is/BusattoLM08} that 
for common XML document trees $t$, almost always the
dag of~$t$ is smaller than the binary dag of~$t$.
One explanation is that $t$ contains many small
repeated subtrees, which seldomly are part of a repeating
sibling end sequence. For each repetition we (possibly) 
pay a ``penalty''
of one extra edge in the dag of $\text{fcns}(t)$; see
the tree $t_n$ which has penalty $n$.
On the other hand, there are very few repeating sibling end
sequences in common XML; this is because optional elements
typically appear towards \emph{the end} of a child sequence.
Hence, the additional feature of sharing sibling sequences
is not useful for XML.
On real XML documents, we show in experiments that the 
``reverse binary dag'' that arises from the 
\emph{last child/previous sibling encoding}
is typically smaller than the binary dag, and almost as small as 
the dag. Moreover, for our test corpus, the average size of the \emph{reverse hybrid dag} built from the 
last child/previous sibling encoding of the dag is only
$62\%$ of the average size of the minimal dag.

Observe that  in the second sharing phase of the construction of the hybrid dag,
only sequences of identifiers (nonterminals of the regular tree grammar corresponding
to the dag) are shared. 
Thus, we are sharing
repeated string suffixes in a sequence of strings.
We experimented with applying a grammar-based string compressor
to this sequence of strings. 
It is not difficult to incorporate the output into an
SLT grammar. 
As our experiments show, the obtained grammars are 
smaller than those of the hybrid dag and almost as small as
TreeRePair's grammars.
Moreover, they  have the advantage that checking equivalence
of subtrees is simple (each distinct subtree is
represented by a unique identifier), a property not present
for arbitrary SLT grammars. 
For hybrid dags, even equality of sibling end sequences can
be checked efficiently.

\noindent
{\bf Average Size Analysis of DAGs.}\quad
Given a tree over $n$ nodes and $m$ labels, what is the average size
of its minimal dag? This problem was studied for unlabeled full
binary trees by Flajolet, Sipala, 
and Steyaert~\cite{FlaSipStey1990}. They present exact expressions and show
that the expected node size of the minimal dag 
of a full binary tree with $n$ nodes is asymptotically
\[
 \kappa \cdot \frac{n}{\sqrt{\ln n}} \cdot \left(1 + O \left(\frac{1}{\ln n} \right) \right)
\]
where the constant $\kappa$ is explicitly determined.
One problem with the paper by Flajolet et.\ al.\
is that the proof of the result above
is rather sketchy, and at certain places contains large gaps. 
Here we fill these gaps, and extend their results, giving detailed proofs of:
\begin{itemize}
\item exact expressions, in terms of their generating functions, 
for the average node and edge sizes of dags and bdags of 
unranked trees over $n$ nodes and $m$ labels, and of
\item the asymptotic behavior of these averages.
We show that these asymptotic behaviors are also of the form 
$C \frac{n}{\sqrt{ \log n}}$,
where $C$ is again explicitly determined.
\end{itemize}
The proofs of these results assume basic knowledge about 
combinatorial classes and generating functions. Details on these
can be found in textbooks, e.g., the one by Flajolet and Sedgewick~\cite{AnalyticCombinatorics}.
Our proofs of the asymptotics are rather involved and can be found in the Appendix.

A preliminary version of this paper (not containing 
average-case sizes of dags) appeared as~\cite{DBLP:conf/icdt/LohreyMN13}.

\section{Trees and dags}\label{sec:trees_and_dags}

Let $\Sigma$ be a finite set of node labels.
An {\em ordered $\Sigma$-labeled multigraph}  is a tuple
$M = (V,\gamma,\lambda)$, where 
\begin{itemize}
\item $V$ is a finite set of nodes
\item $\gamma : V \to V^*$ assigns to each node a finite word over the set of nodes
\item $\lambda : V \to \Sigma$ assigns to each node a label from $\Sigma$.
\end{itemize}
The idea is that for a node $v \in V$, $\gamma(v)$ is the ordered
list of $v$'s successor nodes.
The {\em underlying graph} is the directed graph 
$G_M = (V,E)$, where
$(u,v) \in E$  if and only if  $v$ occurs in  $\gamma(u)$.
The \emph{node size} of $M$, denoted by $ \| M \|$, is 
the cardinality of $V$, 
and the \emph{edge size} or simply \emph{size} of $M$ is 
defined as $|M| = \sum_{v \in V} |\gamma(v)|$ 
(here $|w|$ denotes the length of a word $w$). 
Note that the labeling function $\lambda$ does not influence
the size of $M$. The motivation for this is that
the size of $M$ can be seen as the number of pointers 
that are necessary in order to store $M$ and that 
these pointers mainly determine the space consumption for
$M$. 

\begin{figure*}[t]
\centerline{
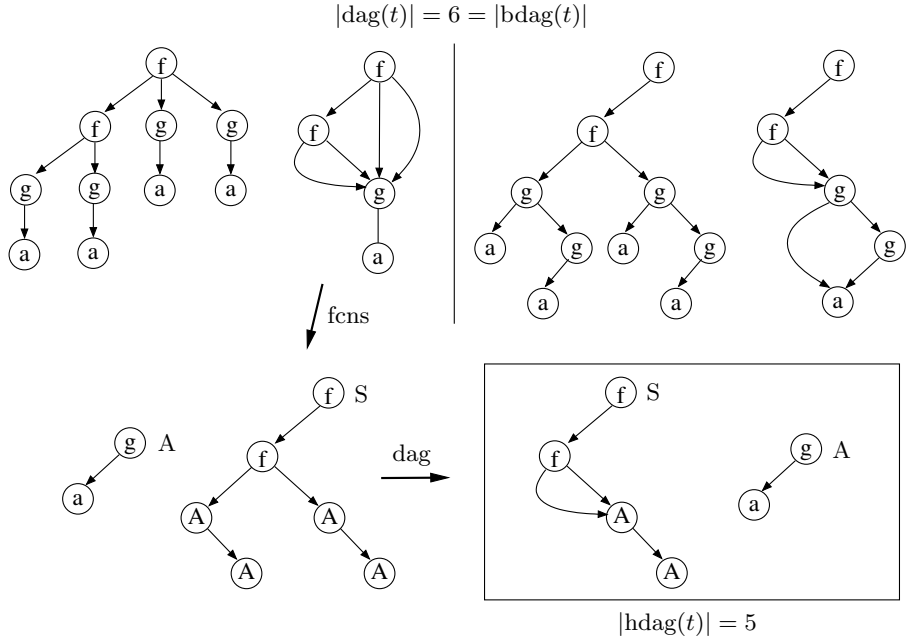
}
\caption{\label{fig:hdag} Top: a tree $t$, its dag, 
its fcns encoding and its bdag of $t$.
Bottom:
its hybrid dag is shown in the box.}
\end{figure*}
Two ordered $\Sigma$-labeled multigraphs 
$M_1 = (V_1,\gamma_1,\lambda_1)$ and $M_2 = (V_2,\gamma_2,\lambda_2)$
are isomorphic if there exists a bijection $f : V_1 \to V_2$ 
such that for all $v \in V_1$, 
$\gamma_2(f(v)) = f(\gamma_1(v))$ and $\lambda_2(f(v)) = \lambda_1(v)$
(here we implicitly extend $f$ to a morphism $f : V_1^* \to V_2^*$).
We do not distinguish between isomorphic multigraphs.
In particular, in our figures a node $v \in V$ is not represented by
the symbol $v$, but by the label $\lambda(v)$.

An {\em ordered $\Sigma$-labeled dag} is a $\Sigma$-labeled  ordered 
multigraph 
$d = (V,\gamma,\lambda)$ 
such that the underlying graph
$G_d$ is acyclic. The nodes $r \in V$ for which 
there is no $v \in V$ such that $(v,r)$ is an edge of 
$G_d$
 ($r$ has 
no incoming edges) are called the {\em roots} of $d$.  
An ordered $\Sigma$-labeled {\em rooted} dag is an ordered $\Sigma$-labeled 
%
dag with a unique root. In this case every node of $d$ 
is
reachable in 
$G_d$ from the root node.
 The nodes $\ell \in V$ for which 
there is no $v \in V$ such that $(\ell,v)$ is an edge of 
$G_d$  ($\ell$ has 
no outgoing edges) are called the {\em leaves} of $d$.
An {\em ordered $\Sigma$-labeled tree} is an ordered $\Sigma$-labeled
rooted dag $t = (V,\gamma,\lambda)$ such that 
every non-root node 
$v$
has exactly one occurrence
in the concatenation of all strings $\gamma(u)$ for $u \in V$.
In other words, the underlying graph 
$G_t$ 
is a rooted
tree in the usual sense and in every string $\gamma(u)$, every
$v \in V$ occurs at most once.
We define $\mathcal{T}(\Sigma)$ as the set of all ordered $\Sigma$-labeled trees. 
We denote ordered $\Sigma$-labeled trees by their usual term notation, i.e.,
for every $a \in \Sigma$, $n \geq 0$, and all trees $t_1,\dots,t_n \in \mathcal{T}(\Sigma)$, 
we also have $a(t_1,\ldots,t_n) \in \mathcal{T}(\Sigma)$.
Note that trees from $\mathcal{T}(\Sigma)$ are \emph{unranked} in the
sense that the number of children of a node does not depend on 
the label of the node.
We therefore frequently speak of unranked trees for
elements of $\mathcal{T}(\Sigma)$.

Let $d= (V,\gamma,\lambda)$ be an ordered $\Sigma$-labeled dag.
With every node $v \in V$ we associate a tree
$\text{eval}_d(v) \in \mathcal{T}(\Sigma)$ 
inductively as follows: We set 
$$\text{eval}_d(v) = f(\text{eval}_d(v_1),\ldots, \text{eval}_d(v_n)),
$$
if $\lambda(v)=f$ and 
$\gamma(v)=v_1\cdots v_n$ (where $f(\varepsilon)= f$).
Intuitively, $\text{eval}_d(v)$ is the tree obtained by unfolding 
$d$ starting in the node $v$. 
If $d$ is an ordered $\Sigma$-labeled rooted dag, then we define 
$\text{eval}(d) = \text{eval}_d(r)$, where $r$ is the root node
of $d$. Note that if $t$ is an ordered 
$\Sigma$-labeled tree and $v$ is a node of $t$, then
$\text{eval}_t(v)$ is simply the subtree of $t$ rooted at $v$
and is written as 
$t/v = \text{eval}_t(v)$ 
in this case.
If for nodes $u \neq v$ of $t$ we have $t/u = t/v$, then
the tree $t/u = t/v$ is a
\emph{repeated subtree} of $t$.

Let $t = (V,\gamma,\lambda) \in \mathcal{T}(\Sigma)$ 
and let $G_t = (V,E)$ be the underlying graph (which is a tree).
For an edge $(u,v) \in E$,  $v$ is a \emph{child} of $u$, and $u$ is
the \emph{parent} of $v$. If two nodes $v$ and $v'$ have the same
parent node $u$, then $v$ and $v'$ are \emph{siblings}.
If moreover $\gamma(u)$ is of the form $u_1 v v' u_2$ for 
$u_1, u_2 \in V^*$ then $v'$ is the \emph{next sibling} of 
$v$, and $v$ is the \emph{previous sibling} of $v'$. 
If a node $v$ does not have a previous sibling, it is a \emph{first child}, and 
if it does not have a next sibling, it is a \emph{last child}.

For many tree-processing formalisms (e.g.\ standard
tree automata), it is useful to deal with ranked trees, where
the number of children of a node is bounded.
There is a standard binary encoding of unranked trees, which we 
introduce next.  
A \emph{binary $\Sigma$-labeled dag} $d$,
or short \emph{binary dag}, can be defined as an ordered $(\Sigma \cup
\{\Box\})$-labeled dag $d = (V,\gamma,\lambda)$, where 
$\Box \not\in \Sigma$ is a special dummy symbol such that the 
following holds: 
\begin{itemize}
\item For every $v \in V$ with $\lambda(v) \in \Sigma$
we have $|\gamma(v)|=2$
\item for every $v \in V$ with $\lambda(v)
=\Box$ we have $|\gamma(v)|=0$. 
\end{itemize}
For a binary dag, $d = (V, \gamma,\lambda)$, we alter our definitions
of node and edge sizes by disregarding all dummy nodes. That is,
the node size
is now
$ \|d\|
= |\{ v \in V \mid \lambda(v) \neq \Box\}|$ 
and the (edge) size is $|d| =  \sum_{v \in V} |\gamma(v)|_{\Sigma}$, where
$|v_1v_2 \cdots v_m|_{\Sigma} = |\{ i \mid 1 \leq i \leq m,
\lambda(v_i) \neq \Box\}|$. 
Accordingly, the dummy nodes are not represented in our figures.

A \emph{binary $\Sigma$-labeled tree} $t$,
or short \emph{binary tree}, is a binary dag which is moreover
an ordered $(\Sigma \cup \{\Box\})$-labeled tree.
The symbol $\Box$ basically denotes the absence
of a left or right child of a node. For instance, 
$g(a,\Box)$ denotes the binary tree that has a $g$-labeled root with an
$a$-labeled left child but no left child
(as shown in the bottom of Figure~\ref{fig:hdag}).
Note that $g(a,\Box)$ and $g(a,\Box)$ are different binary trees.
Let $\mathcal{B}(\Sigma)$ denote the set of 
binary $\Sigma$-labeled trees.

We define a mapping 
$\text{fcns} : \mathcal{T}(\Sigma)^* \rightarrow \mathcal{B}(\Sigma)$,
where as usual $\mathcal{T}(\Sigma)^*$ denotes the set of all finite
words (or sequences) over the set $\mathcal{T}(\Sigma)$,
as follows (``fcns'' refers to ``first child/next sibling''): 
For the empty word $\varepsilon$ let
$\text{fcns}(\varepsilon) = \Box $ (the empty binary tree).
If $n \geq 1$, 
$t_1, \ldots, t_n \in \mathcal{T}(\Sigma)$ 
and $t_1 = f(u_1,\ldots,u_m)$ with $m \geq 0$, then
$$
\text{fcns}(t_1 t_2\cdots t_n) = f(\text{fcns}(u_1 \cdots
u_m),\text{fcns}(t_2 \cdots t_n)).
$$
Note that  $\text{fcns}(a) = a(\Box,\Box)$ for $a \in \Sigma$.
In the following we always simply write $a$ for $a(\Box,\Box)$. 
The encoding fcns is bijective, hence the inverse 
$\text{fcns}^{-1} : \mathcal{B}(\Sigma) \to \mathcal{T}(\Sigma)^*$ 
is defined. Moreover, for every $t \in \mathcal{T}(\Sigma)$,
we have
 $\|\text{fcns}(t)\| =\|t\|$,
%
see e.g.~\cite{DBLP:books/aw/Knuth68}.
The fcns encoding is also known as 
rotation correspondence (see, e.g.~\cite{marckert}) 
and as Rabin encoding.

\begin{example} 
Let $t_1 = f(a_1,a_2,a_3)$ and let $t_2 = g(b_1,b_2)$. Then 
$\text{fcns}(t_1 t_2) = f(a_1(\Box,a_2(\Box,a_3)),g(b_1(\Box,b_2),\Box))$.
\end{example}
As mentioned in the Introduction, 
one can construct $\text{fcns}(t)$ by keeping
all nodes of $t$ and 
creating
 edges as follows: For each node $u$
of $t$, the left child of $u$ 
in $\text{fcns}(t)$
is the first child of $u$ in $t$ (if it
exists) and the right child of $u$ 
in $\text{fcns}(t)$
is the next sibling of $u$ in $t$ (if it exists).

An ordered tree can be compacted by representing occurrences of repeated
subtrees only once. Several edges then point to the same subtree
(which we call a \emph{repeated} subtree), thus making the tree 
an ordered dag. It is known that the minimal
dag of a tree is unique and that it can be constructed in linear time 
(see e.g., ~\cite{DBLP:journals/jacm/DowneyST80}). 
For later purposes it is useful to define 
the minimal dag $\text{dag}(d)$ for every ordered $\Sigma$-labeled dag
$d= (V,\gamma,\lambda)$. 
It can be 
%
 defined as 
$$
\text{dag}(d) = (\{ \text{eval}_d(u) \mid u \in V\}, \gamma',\lambda')
$$
with $\lambda'(f(t_1,\ldots, t_n)) =  f$ and 
$\gamma'(f(t_1,\ldots, t_n)) = t_1\cdots t_n$.
Thus, the nodes of $\text{dag}(d)$ are the different trees represented
by the 
unfoldings of the
nodes of $d$.

%
The internal structure of the nodes of $\text{dag}(d)$ 
(which are trees in our definition) has no influence on the size of $\text{dag}(d)$,
which is still defined to be the number of its edges.
Note that in general we cannot recover
$d$ from $\text{dag}(d)$: For instance if $d$ is the disjoint union of
two copies of the same tree $t$, then $\text{dag}(d) = \text{dag}(t)$,
but this will not be a problem. 
Indeed, 
we use dags only for the compression
of forests consisting of different trees.
Such a forest can be reconstructed from its minimal dag.
Note also that if $d$ is a rooted dag, then $\text{dag}(d)$ is also rooted 
and we have $\text{eval}(\text{dag}(d))=\text{eval}(d)$.

\begin{example} Consider the 
tree $t_n$ 
defined by $t_0=a$ and $t_n= a(t_{n-1},t_{n-1})$.
While 
$|t_n| = 2(2^n-1)$, $|\text{dag}(t_n)| =
2n$. Hence $\text{dag}(t)$ can be exponentially smaller than $t$. 
\end{example}
The \emph{binary dag} of $t\in \mathcal{T}(\Sigma)$, denoted $\text{bdag}(t)$, is defined as
$$
\text{bdag}(t)  = \text{dag}(\text{fcns}(t)).
$$
It is a binary dag as defined above. See Figure~\ref{fig:tn2}
in the Introduction for an example 
(recall that we do not represent dummy nodes in binary dags).

Clearly, the number of nodes of $\text{dag}(t)$ equals 
the number of different subtrees $t/v$ 
of $t$.
In order to 
describe
 the number of nodes of $\text{bdag}(t)$ the
following definition is useful: For a node $v$ of an unranked tree 
$t = (V,\gamma,\lambda)$ define
$\text{sibseq}(v) \in \mathcal{T}(\Sigma)^*$ 
(the {\em sibling sequence} of $v$) as 
the sequence of subtrees rooted at $v$ and all its
right siblings. More formally,
if $v$ is the root of $t$ then $\text{sibseq}(v)=t$. Otherwise,
let $u$ be the parent node of $v$ and let 
$\gamma(u) = w v v_1 \cdots v_m$, where $w \in V^*$.
Then  
$$
\text{sibseq}(v) = (t/v) (t/v_1) \cdots (t/v_m).
$$

\begin{example} 
The different sibling sequences of the tree $t=f(a, f(b, a), b, a)$ are: $t$,
$af(b,a)ba$, $f(b,a)ba$, $ba$, and $a$.
\end{example} 
The following lemma follows directly from the definitions
of bdag and sibseq:

\begin{lemma} \label{lemma-sib-sequ}
The number of nodes of $\text{bdag}(t)$ is equal to the number of different
sibling sequences $\text{sibseq}(v)$, for all $v \in V$.
\end{lemma}

\section{Straight-line tree grammars} \label{sec-SLT-grammar}


Straight-line tree grammars are a formalism that in many cases 
give a more compact tree representation as dags. 
Let $\{y_1, y_2, \ldots\}$ be an infinite fixed set of parameters (or variables).
A \emph{straight-line tree grammar} (\emph{SLT grammar} for short)
is a tuple ${\mathcal G}=(N, \text{rank}, \Sigma, S, \rho)$, where 
\begin{itemize}
\item $N$ is a finite set of so-called {\em nonterminal symbols}
\item $\text{rank} : N \to \mathbb{N}$ maps every nonterminal
to its rank 
(which may be~$0$)
\item $\Sigma$ is a finite set of node labels
\item $S \in N$ is the start nonterminal and 
\item $\rho$ is a function that maps every $X \in N$  to an ordered 
$\Gamma$-labeled tree $\rho(X) = (V,\gamma,\lambda)$, where
$\Gamma=\Sigma \cup N \cup \{y_1, \ldots, y_{\text{rank}(X)}\}$ and
the following conditions hold:
%
\begin{itemize}
\item for every $1 \leq i \leq \text{rank}(X)$ 
there is exactly one node $v \in V$ with $\lambda(v) = y_i$, which
moreover is a leaf of $\rho(X)$ and
\item  for every
node $v \in V$ with $\lambda(v) = Y \in N$ we have $|\gamma(v)| = \text{rank}(Y)$, i.e., $v$ has $\text{rank}(Y)$
many children. 
\end{itemize}
Finally, the binary relation
$\{ (X,Y) \in N \times N \mid Y \text{ appears in } \rho(X) \}$ must be acyclic. 
\end{itemize}
 We also write $X \to t$
for $\rho(X)=t$ and call $X \to t$ a \emph{rule} or \emph{production} of $\mathcal G$. Moreover, we also write 
$X(y_1, \ldots, y_{\text{rank}(X)})$ instead of  $X$
in the left-hand side of the rules, to emphasize the rank of the
nonterminals.

The 
properties
of an SLT grammar ${\mathcal G}=(N, \text{rank}, \Sigma, S, \rho)$
allow us to define 
for every nonterminal $X \in N$ a $(\Sigma \cup \{y_1, \ldots,
y_{\text{rank}(X)}\})$-labeled 
tree $\text{eval}_{\mathcal G}(X)$ inductively 
as follows\footnote{We hope that no confusion will arise with the evaluation of a dag
defined in the previous section; we will see in fact that the
evaluation of a dag can be seen as a special case of the evaluation of
an SLT grammar.}:
 Let $\rho(X) = (V,\gamma,\lambda)$.
Assume that for every nonterminal $Y$ that appears in 
$\rho(X)$, the tree $t_Y = \text{eval}_{\mathcal G}(Y)$ is already defined.
This is a tree that contains for every $1 \leq j \leq \text{rank}(Y)$
exactly one leaf node that is labeled with $y_j$. We now replace every node $v \in V$ in $\rho(X)$
that is labeled with a nonterminal by a copy of the tree $t_{\lambda(v)}$.
Thereby, the $j$-th child of $v$ is identified
with the $y_j$-labeled node of $t_{\lambda(v)}$ for every
$1 \leq j \leq \text{rank}(\lambda(v))$, and the parent node of the root of $t_{\lambda(v)}$ becomes the parent
node of $v$. The resulting tree is $\text{eval}_{\mathcal G}(X)$.
For a completely formal definition, see e.g. 
\cite{DBLP:journals/tcs/LohreyM06,DBLP:journals/jcss/LohreyMS12}.\footnote{The formalisms in 
\cite{DBLP:journals/tcs/LohreyM06,DBLP:journals/jcss/LohreyMS12} slightly differ from our definition,
since they assume a fixed rank for every node label in $\Sigma$. But this is not an essential difference.}
Finally, let $\text{eval}({\mathcal G}) = \text{eval}_{\mathcal G}(S)$.
The term ``straight-line tree grammar''  comes from the fact that an SLT can be seen as a context-free
tree grammar that produces a single tree.

The size of ${\mathcal G}=(N, \text{rank}, \Sigma, S, \rho)$ is 
defined to be
$$
|{\mathcal G}| = \sum_{X \in N} |\rho(X)| .
$$

\begin{example}
Consider the SLT grammar 
$\mathcal G$ 
with nonterminals $S, A, B$ and  rules
\begin{eqnarray*}
S & \to & B(a,b,B(c,d,a)),  \\
B(y_1, y_2, y_3) & \to & A(y_1, A(y_2,y_3)), \\
A(y_1, y_2) & \to & f(g(y_1),y_2) .
\end{eqnarray*}
We have $|{\mathcal G}| = 13$ and
$\text{eval}({\mathcal G}) = f(g(a),f(g(b),f(g(c),f(g(d),a))))$.
The same tree is also generated by the SLT grammar
with the rules
\begin{eqnarray*}
S & \to & A(a,A(b,A(c,A(d,a)))),  \\
A(y_1, y_2) & \to & f(g(y_1),y_2) .
\end{eqnarray*}
Its size is only $11$.
\end{example}

A $k$-SLT is an SLT ${\mathcal G} = (N, \text{rank}, \Sigma, S, \rho)$ such that
$\text{rank}(X) \leq k$ for every $X \in N$.  
A $0$-SLT grammar is
also called a \emph{regular SLT} (since it is a regular tree
grammar). 
In such a grammar, nonterminals only occur as leaves in the right-hand
sides of the rules.

An ordered $\Sigma$-labeled rooted dag $d = (V, \gamma,\lambda)$ can be
identified with the $0$-SLT grammar 
$$
\mathcal{G}_d = (V,\text{rank},\Sigma,S,\rho),
$$ 
where  $\text{rank}(v)=0$ for every $v \in V$, $S$ is the root of $d$,
and $\rho(v) = f(v_1,\ldots,v_n)$ if $\lambda(v)=f$ and $\gamma(v) =
v_1 \cdots v_n$. 
Note that all trees occurring in the right-hand sides of the rules have
height $0$ or $1$.
Often in this paper, it will be useful to eliminate in
$\mathcal{G}_d$ nonterminals $v \in V$ such that 
$\gamma(v) = \varepsilon$
(that is, the leaves of the dag $d$).
For this, we define the {\em reduced}  $0$-SLT grammar 
$$
\mathcal{G}_d^{\text{red}} = (V \setminus V_0 ,\text{rank},\Sigma,S,\rho),
$$ 
where $V_0 = \{ v \in V \mid \gamma(v) = \varepsilon\}$,  
$\text{rank}(v)=0$ for every $v \in V \setminus V_0$, $S$ is the root of $d$,
and $\rho(v) = f(w_1,\ldots,w_n)$ if $\lambda(v)=f$, $\gamma(v) = v_1 \cdots v_n$, and
$w_i = \lambda(v_i)$ if $v_i \in V_0$ and $w_i = v_i$ otherwise.
Note that every right-hand side $\rho(v)$ of
$\mathcal{G}_d^{\text{red}}$ is now a tree of height 1.
This
does not change the evaluation of the grammar,
and simplifies some technical details in Section~\ref{sec-size}.
Of course, we should exclude the case that $V_0 = V$. For this, we simply
exclude 
dags
 consisting of a single node 
 from further considerations.
Finally, observe that the evaluation of the grammar $\mathcal G_d$
coincides with the evaluation of the dag $d$ defined in the previous section.

If $\mathcal{G}_d^{\text{red}} = (V \setminus V_0 ,\text{rank},\Sigma,S,\rho)$  with $V \setminus V_0 = \{ A_1, \ldots, A_n\}$ 
and $\rho(A_i) = f_i(\alpha_{i,1}, \ldots, \alpha_{i,k_i})$ (where $\alpha_{i,j} \in (V \setminus V_0) \cap \Sigma$)
then the words $\alpha_{i,1} \cdots \alpha_{i,k_i} \in   ((V \setminus V_0) \cap \Sigma)^+$ are called the \emph{child
sequences} of the dag $d$. 

\begin{example}
For $d = \text{dag}(t)$ in Figure~\ref{fig:hdag}, $\mathcal{G}_d^{\text{red}}$ consists of the rules
$$
S \to f(B,A,A), \qquad B \to f(A,A), \quad A \to g(a).
$$
Hence the child sequences of the dag are $BAA$, $AA$, and $a$.
\end{example}
Algorithmic problems on SLT grammar-compressed trees were considered in \cite{DBLP:journals/tcs/LohreyM06,DBLP:journals/jcss/LohreyMS12}. 
Of particular interest
in this context is a result from \cite{DBLP:journals/jcss/LohreyMS12} stating that a given SLT $\mathcal G$ can be transformed in polynomial time
into an equivalent $1$-SLT ${\mathcal G}_1$  such that 
$\text{eval}({\mathcal G}_1)=\text{eval}({\mathcal G})$.  In combination with results from 
\cite{DBLP:journals/tcs/LohreyM06} it follows that for a given nondeterministic  tree automaton $A$ (even with sibling constraints) 
and an SLT grammar $\mathcal G$ one can test in polynomial time whether $A$ accepts $\text{eval}({\mathcal G})$.
Compression algorithms that produce from a given input tree a small SLT grammar are proposed in 
\cite{DBLP:journals/is/BusattoLM08,lohmanmen13}.

\section{The hybrid dag} \label{sec-hdag}
While the dag shares repeated subtrees of a tree, 
the binary dag shares repeated sibling sequences (see Lemma~\ref{lemma-sib-sequ}). 
Consider an unranked tree $t$.
As we have seen in the Introduction, the size of
$\text{dag}(t)$ can be smaller than the size of
$\text{bdag}(t)$. On the other hand, it can also be that the size of
$\text{bdag}(t)$ is smaller than the size of $\text{dag}(t)$.
We now wish to define a tree representation that combines both types
of sharing (trees and tree sequences) and whose size is
guaranteed to be smaller than or equal to the minimum
of the sizes of $\text{dag}(t)$ and $\text{bdag}(t)$.
Our starting point is $d = \text{dag}(t)$. 
In this dag we now want to share all repeated sibling sequences.
As an example, consider the tree
$t=f(f(g(a),g(a)),g(a),g(a))$ shown 
on the top left of Figure~\ref{fig:hdag}. 
Its size is $9$.
The dag of this tree
consists of a unique occurrence of the subtree
$g(a)$ plus two $f$-labeled nodes, shown
to the right of $t$ in the figure. 
Thus $|d|=6$.
The corresponding reduced $0$-SLT grammar  $\mathcal{G}_d^{\text{red}}$ consists of the rules
\begin{equation}\label{dag-example}
\begin{array}{lcl}
S&\to& f(B,A,A),\\
B &\to& f(A,A), \\
A &\to& g(a).
\end{array}
\end{equation}
In order to share repeated sibling sequences in $d$
we apply the fcns encoding to the right-hand sides of
$\mathcal{G}_d^{\text{red}}$. 
For the above example we obtain the following new ``binary tree grammar'':
\begin{equation}\label{dag-example-binary-coding}
\begin{array}{lcl}
S&\to& f(B(\Box, A(\Box,A)),\Box)\\
B &\to& f(A(\Box,A),\Box)\\
A&\to& g(a,\Box).
\end{array}
\end{equation}
This is not an SLT grammar, since there are occurrences of $A$ 
with $0$ and $2$, respectively,
children. We view the above rules just as the binary encoding of~\eqref{dag-example}.

We now build the minimal dag of the 
forest
obtained by taking the disjoint union
of all right-hand sides of \eqref{dag-example-binary-coding}.
In the example, the subtree
$A(\Box,A)$ appears twice and is shared.
We write the resulting dag again as a grammar,
using the new nonterminal $C$ for the new repeated tree $A(\Box,A)$
(corresponding to the repeated sibling sequence $AA$ in \eqref{dag-example}):
\begin{equation}\label{hdag-example}
\begin{array}{lcl}
S&\to& f(B(\Box, C),\Box)\\
B &\to& f(C,\Box)\\
C & \to & A(\Box,A)\\
A&\to& g(a,\Box)
\end{array}
\end{equation}
These rules make up  the {\em hybrid dag} ({\em hdag} for short)
of the initial tree. Its size is the total number of edges in all
right-hand side trees; it is $5$ in our example (here, as usual, we do not count edges
to $\Box$-labeled nodes).
Compare this to $\text{dag}(t)$ and
$\text{bdag}(t)$, which are both of size $6$. Again, note that \eqref{hdag-example}
should not be seen as an SLT-grammar but as a succinct encoding of \eqref{dag-example}.

In our example, the production $B \to f(A,A)$ in the $0$-SLT grammar \eqref{dag-example}
does not save any edges, since the nonterminal $B$ occurs only once in a right-hand side (namely $f(B,A,A)$).
Eliminating this production yields the $0$-SLT grammar
\[
\begin{array}{lcl}
S&\to& f(f(A,A),A,A)\\
A &\to& g(a)
\end{array}
\]
with the fcns encoding
\[
\begin{array}{lcl}
S&\to& f(f(A(\Box,A),A(\Box,A)),\Box)\\
A&\to& g(a,\Box) .
\end{array}
\]
Sharing repeated subtrees gives 
\begin{equation}\label{hdag2-example}
\begin{array}{lcl}
S&\to& f(f(C,C),\Box)\\
C&\to& A(\Box,A)\\
A&\to& g(a,\Box) ,
\end{array}
\end{equation}
which corresponds to the framed graph in Figure~\ref{fig:hdag} .
The total number of edges to  non-$\Box$ nodes in all right-hand sides is still 5, but it has only 
3 nonterminals in contrast to 4 for the above hdag. In practice,
having fewer nonterminals is preferable. In fact,
our implementation avoids redundant nonterminals like $B$ in our 
example. On the other hand, having only trees of height 1 as right-hand
sides of the dag (seen as a reduced $0$-SLT grammar) does not influence the number of edges
in the final grammar. Moreover, it slightly simplifies the proofs in the next section, where
we show that the size of the hdag of a tree $t$
is smaller than or equal to the minimum of the sizes of
$\text{dag}(t)$ and $\text{bdag}(t)$. 

In general,
the hybrid dag is produced by first building the minimal dag, then
constructing the fcns encoding of the corresponding reduced 
$0$-SLT grammar, and then
building a minimal dag again.
More formally, consider $d = \text{dag}(t)$ and assume that
the corresponding reduced $0$-SLT grammar $\mathcal{G}_d^{\text{red}}$
contains the rules
$A_1\to t_1,\dots, A_n\to t_n$. Recall that every tree $t_i$ has height
1 and that the trees $t_1, \ldots, t_n$ are pairwise different.
Let $t'_i$ be the tree that is obtained from $t_i$ by adding $A_i$ as an additional
label to the root of $t_i$.
Then 
\[
\text{hdag}(t) = \text{dag}(\text{fcns}(t'_1),\dots,\text{fcns}(t'_n)),
\]
where the tuple $(\text{fcns}(t'_1),\dots,\text{fcns}(t'_n))$
is viewed as the dag obtained by taking the disjoint union
of the binary trees $\text{fcns}(t'_i)$.
Clearly $\text{hdag}(t)$ is unique up to isomorphism.
In the second step when 
$\text{dag}(\text{fcns}(t'_1),\dots,\text{fcns}(t'_n))$
is constructed from the tuple 
$(\text{fcns}(t_1'),\dots,\text{fcns}(t_n'))$, 
only suffixes of child sequences can be shared,
since the trees $t'_1, \ldots, t'_n$ are pairwise different and of height 1.
The size $|\text{hdag}(t)|$ of the hdag is the number of edges (to non-$\Box$-labeled nodes) 
of $\text{dag}(\text{fcns}(t'_1),\dots,\text{fcns}(t'_n))$.
Note that the additional
label $A_i$ at the root of $t_i$ is needed in order to be able to reconstruct
the initial tree $t$. In \eqref{hdag-example}, these additional labels ($S$, $A$, and $B$) are implicitly present as the left-hand sides of the rules.
On the other hand, these labels have no influence on the size of the hdag.

The hdag is a particular dag. It is obtained by sharing
repeated suffixes of 
child sequences in the minimal dag (viewed as a $0$-SLT grammar).
In Section~\ref{sec:dag_plus_string} we introduce a further generalization of this
idea, where child sequences of the dag are compressed using a straight-line 
context-free tree grammar.  Moreover, we show that such a compressed structure
can be easily transformed into an equivalent $1$-SLT grammar (Theorem~\ref{thm-construct-1-STL}). This applies 
in particular to hdags. Hence, all the good algorithmic properties of ($1$-)SLT grammars
(e.g. polynomial time evaluation of tree automata) also hold for hdags.

\section{Using the reverse encoding}\label{section_reverseDag}
Instead of using the fcns encoding of a tree, one may also use 
the \emph{last child previous sibling encoding} (lcps). 
Just like fcns, lcps is a bijection from
$\mathcal{T}(\Sigma)^*$ to $\mathcal{B}(\Sigma)$ 
and is defined as follows.
For the empty word $\varepsilon$ let
$\text{lcps}(\varepsilon) = \Box $ (the empty binary tree).
If $n \geq 1$, 
$t_1, \ldots, t_n \in \mathcal{T}(\Sigma)$ 
and $t_n = f(u_1,\ldots,u_m)$ with $m \geq 0$, then
$$\text{lcps}(t_1 t_2\cdots t_n) = f(\text{lcps}(t_1,\ldots,t_{n-1}),
%
  \text{lcps}(u_1 \cdots u_m)).
$$
Again, the inverse $\text{lcps}^{-1} : \mathcal{B}(\Sigma)
\to \mathcal{T}(\Sigma)^*$ is defined.

\begin{example}
Let $t_1=f(a_1,a_2,a_3)$ and  $t_2 = g(b_1,b_2)$. 
Then $$\text{lcps}(t_1 t_2) = g(f(\Box,a_3(a_2(a_1,\Box),\Box)),b_2(b_1,\Box)).$$
\end{example}
Let $\text{rbdag}(t)=\text{dag}(\text{lcps}(t))$ 
and $$\text{rhdag}(t)=\text{dag}(\text{lcps}(t'_1),\dots,\text{lcps}(t'_n)),$$
where $t'_1, \ldots, t'_n$ are obtained from $t$ as in the definition of the hdag.
The reason to consider the lcps encoding is that 
$\text{rbdag}(t)$ and $\text{rhdag}(t)$ are smaller for trees that have 
repeated \emph{prefixes} of child sequences. 
Empirically, as we show in Section~\ref{sec:exp_results_dag}, this is 
quite common and for most trees $t$ in our 
XML corpus $|\text{rbdag}(t)| < |\text{bdag}(t)| $ and 
$|\text{rhdag}(t)| < |\text{hdag}(t)| $.

\begin{example}   
Let $t = f(f(a,a,b),f(a,a,c))$. 
Then $|\text{rbdag}(t)| = 7$ while 
$|\text{dag}(t)|=|\text{bdag}(t)|=|\text{hdag}(t)|=|t|=8$.
\end{example}
Clearly, there are also trees $t$ where $|\text{hdag}(t)| < |\text{rhdag}(t)|$.
This raises the question whether there is a scheme which 
combines the best of both approaches.
Obviously one can construct both $\text{hdag}(t)$ and $\text{rhdag}(t)$
of a tree $t$ and discard the larger of both. Yet a scheme which combines 
both approaches by sharing both suffixes and prefixes of children sequences,
faces the problem that the resulting minimal object is not necessarily unique. 
This can easily be seen by considering trees in which repeated prefixes and suffixes
of child  sequences overlap. Also it is not clear how a minimal such object can 
be constructed efficiently.
A (non-optimal) approach we have considered
was to first share repeated prefixes and 
then share repeated suffixes. Yet the results in compression achieved 
were not significantly better than for the $\text{rhdag}$. Moreover,
this approach can be further generalized by sharing arbitrary factors of  
sibling sequences. This is the topic of Section~\ref{sec:dag_plus_string}.

\section{Comparison of worst-case sizes of dag, bdag, and hdag} \label{sec-size}

We want to compare the node size and the edge size of $\text{dag}(t)$, $\text{bdag}(t)$, and $\text{hdag}(t)$ for an unranked tree $t$.
We do not include $\text{rbdag}(t)$ or $\text{rhdag}(t)$, 
because by symmetry the same bounds holds 
as for $\text{bdag}(t)$ and $\text{hdag}(t)$, respectively.
   
\subsection{The number of nodes}

   In this section we consider the number of \emph{nodes}
   in the dag and bdag of an unranked tree $t$. We show that
   $\| \text{dag}(t)\| \leq \|\text{bdag}(t)\|$.

   \begin{example} \label{example-f(a,..a)}
   Consider the tree $t_n=f(a,a,\ldots,a)$ 
   consisting of $n$ nodes, where $n \geq 2$.
   Then 
$\|\text{dag}(t)\|=2$ 
and
 $\|\text{bdag}(t)\|=n$, while $|\text{dag}(t)| = |\text{bdag}(t)|
   = n-1$.
   Note that dags with multiplicities on edges, as defined 
   in~\cite{DBLP:conf/vldb/KochBG03}, can store a tree such
   as $t_n$ in size $O(\log n)$. 
   \end{example}

   \begin{lemma}\label{lemma:comparing_node_size}
    Let $t$ be an unranked tree. Then
$\| \text{dag}(t)\| \le 
\| \text{bdag}(t) \| \label{ineq:node_size}$.
   \end{lemma}

    \begin{proof}
    The lemma follows from Lemma~\ref{lemma-sib-sequ} and the obvious
    fact that the number of different subtrees of $t$ (i.e., $
\|\text{dag}(t)\|$)
    is at most the
    number of different sibling sequences in $t$: $\text{sibseq}(u) =
    \text{sibseq}(v)$ implies  $t/u = t/v$.
    \qed
    \end{proof}

   \begin{lemma} \label{lemma:comparing_node_size2}
    There exists a family of trees 
    $(t_n)_{n \geq 2}$ such that $
\|\text{dag}(t)\| = 2$ and
    $
\|t_n\| = 
\|\text{bdag}(t)\| = n$.
   \end{lemma}
   \begin{proof}
    Take the family of trees $t_n$ from Example~\ref{example-f(a,..a)}.
    \qed
   \end{proof}
Let us remark that the node size of the hdag can be larger than the node size of 
the bdag and the node size of the dag. The reason is that
in $\mathcal{G}^{\text{red}}_{\text{dag}(t)}$, there is a nonterminal for each node
of the dag (and hence the height of each right-hand side is at most one).
This can be done differently of course; it was chosen to simplify proofs and 
because our main goal is the reduction of edge size. 
Note that the total number of edges of $\mathcal{G}^{\text{red}}_{\text{dag}(t)}$
is equal to the number of edges of $\text{dag}(t)$. 

\subsection{The number of edges} \label{sec:number-edges}

We have just seen 
that the number of nodes of the
(minimal) dag is always at most the number of nodes of the bdag,
and that the gap can be maximal ($O(1)$ versus $|t|$). For the number
of edges, the situation is different. We show that
$\frac{1}{2} |\text{bdag}(t)| \le |\text{dag}(t)| \leq
\frac{1}{2} |\text{bdag}(t)|^2$ for $|t| \geq 2$ and that these bounds are sharp up to the constant
factor $1/2$ in the second inequality.
In fact, for $|t| \geq 2$ we show the three inequalities
\begin{eqnarray*}
|\text{hdag}(t)| & \le & \min( |\text{dag}(t)|, |\text{bdag}(t)|), \\
|\text{bdag}(t)|  & \le & 2|\text{hdag}(t)|, \text{ and} \\
|\text{dag}(t)| & \leq & \frac{1}{2} |\text{hdag}(t)|^2
\end{eqnarray*}
which imply 
$$\frac{1}{2} |\text{bdag}(t)| \le |\text{dag}(t)| \leq
\frac{1}{2} |\text{bdag}(t)|^2.
$$
Before we prove these bounds we need some definitions.
Recall that the nodes of $\text{bdag}(t)$ are in 1-1-correspondence
with the different sibling sequences of $t$. In the following, let
$$
\text{sib}(t) = \{\text{sibseq}(v) \mid v \text{ a node of } t\}
$$
be the set of all sibling sequences of $t$.
To count the size (number of edges) of $\text{bdag}(t)$
we have to count for each sibling sequence $w \in \text{sib}(t)$
the number of outgoing edges in $\text{bdag}(t)$. We denote
this number with $e(w)$; it can be computed as follows,
where $w = s_1 s_2 \cdots s_m$ ($m \geq 1$) and the $s_i$ are trees:
\begin{itemize}
\item  
$e(w) = 0$, if $m=1$ and $|s_1|=0$
\item  
$e(w) = 1$, if either $m=1$ and $|s_1| \geq 1$ (then $w$ has
  only a left child) or if $m \geq 2$ and $|s_1|=0$ (then $w$ has only a right child)
\item $e(w) = 2$, otherwise.
\end{itemize}
With this definition we obtain:

\begin{lemma} \label{lemma-counting-edges-bdag}
For every $t \in \mathcal{T}(\Sigma)$, we have
$$
|\text{bdag}(t)| = \sum_{w \in \text{sib}(t)} e(w).
$$
\end{lemma}
The size of the hdag can be computed similarly: 
Consider  the reduced $0$-SLT grammar $\mathcal{G} = \mathcal{G}_{\text{dag}(t)}^{\text{red}}$. Let $N$ be 
the set of nonterminals of $\mathcal{G}$ and let $S$ be the start
nonterminal. Recall that every right-hand
side of $\mathcal{G}$ has the form $f(\alpha_1,\ldots, \alpha_n)$,
where every $\alpha_i$ belongs to $\Sigma \cup N$.
Let $\text{sib}(\mathcal{G})$ be the set of all sibling sequences that occur in 
the right-hand sides of $\mathcal{G}$.
Thus, for every right-hand side $f(\alpha_1,\ldots,\alpha_n)$ of 
$\mathcal{G}$, the sibling sequences 
$f(\alpha_1,\ldots,\alpha_n)$ (a sibling sequence of length $1$)
and $\alpha_i \alpha_{i+1} \cdots \alpha_n$ ($1 \leq i \leq n$) 
belong to $\text{sib}(\mathcal{G})$. For such a sibling sequence
$w$ we define $e(w)$ as above. Here, every $\alpha_i$ is viewed
as a tree with a single nodes, i.e., $|\alpha_i|=0$.
Then we have:

\begin{lemma} \label{lemma-counting-edges-hdag}
For every $t \in \mathcal{T}(\Sigma)$, we have
$$
|\text{hdag}(t)| = \sum_{w \in \text{sib}(\mathcal{G})} e(w).
$$
\end{lemma}
For $w = s_1\cdots s_m \in \text{sib}(t)$ let $\tilde{w}$ be the string that 
results from $w$ by replacing every non-singleton tree $s_i
\not\in\Sigma$ by the unique nonterminal of $\mathcal{G}$ that
derives to $s_i$. Actually, we should write $\tilde{w}_t$ instead of 
$\tilde{w}$, since the latter also depends on the tree $t$. But the
tree $t$ will be always clear from the context.
Here are a few simple statements:
\begin{itemize}
\item
For every $w \in \text{sib}(t)$, the sibling sequence $\tilde{w}$ 
belongs to $\text{sib}(\mathcal{G})$, except for the length-1 sequence $\tilde{w} = S$ that is obtained from  
the length-1 sequence $w = t \in \text{sib}(t)$. 
\item For every $w \in \text{sib}(t)$, $\tilde{w}$ is a word over $N \cup \Sigma$.
\item For every $w \in \text{sib}(t)$, $e(\tilde{w}) \leq e(w)$.
\item The mapping $w \mapsto \tilde{w}$ is an injective
mapping from $\text{sib}(t)  \setminus \{t\}$ to $\text{sib}(\mathcal{G})$.
\end{itemize}
Using this mapping, 
the sums in Lemma~\ref{lemma-counting-edges-bdag} and
\ref{lemma-counting-edges-hdag} can be related
as follows:

\begin{lemma} \label{lemma-relating-sums}
For every $t \in \mathcal{T}(\Sigma)$, we have
$$
\text{hdag}(t) \ =  \sum_{w \in \text{sib}(\mathcal{G})} \!\!\!\! e(w)
\ = \ |N| + \sum_{w \in \text{sib}(t)} e(\tilde{w}) .
$$
\end{lemma}

\begin{proof}
By Lemma~\ref{lemma-counting-edges-hdag} it remains to show
the second equality.
The only  sibling sequences in
$\text{sib}(\mathcal{G})$ that are not of the form $\tilde{w}$ for $w \in \text{sib}(t)$ are the sequences (of length 1) 
that consist of the whole right-hand side $f(\alpha_1,\ldots,\alpha_m)$
of a nonterminal $A \in N$.
For such a sibling sequence $u$ we have $e(u) = 1$ (since it has length $1$ and 
$f(\alpha_1,\ldots,\alpha_m)$ is not a single symbol).
Hence, we have 
\begin{eqnarray*}
\sum_{w \in \text{sib}(\mathcal{G})} \!\!\!\!  e(w) & = &  
|N| + \!\!  \sum_{w \in \text{sib}(t) \setminus \{t\}} \!\!\!\!  e(\tilde{w})
\\
& = & |N| + \sum_{w \in \text{sib}(t)} e(\tilde{w}) ,
\end{eqnarray*}
where the last equality follows from $e(\tilde{t}) = e(S) = 0$.
\qed
\end{proof}

\begin{theorem} \label{hdag-kleiner-min}
For every $t \in \mathcal{T}(\Sigma)$, we have
$$
|\text{hdag}(t)| \le \min( |\text{dag}(t)|, |\text{bdag}(t)|).
$$
\end{theorem}

\begin{proof}
Since $\text{hdag}(t)$ is obtained from $\text{dag}(t)$ by sharing
repeated suffixes of child sequences, we immediately get
$|\text{hdag}(t)| \le |\text{dag}(t)|$. 
It remains to show $|\text{hdag}(t)| \leq |\text{bdag}(t)|$.
By Lemma~\ref{lemma-counting-edges-bdag} and \ref{lemma-relating-sums}
we have to show
$$|N|+\sum_{w \in \text{sib}(t)} e(\tilde{w}) \ \leq \ \sum_{w \in \text{sib}(t)} e(w),$$
where $N$ is the set of nonterminals of $\mathcal{G}_{\text{dag}(t)}^{\text{red}}$.
To see this, note that: 
\begin{itemize}
\item $e(\tilde{w}) \leq e(w)$ for all $w \in \text{sib}(t)$ and 
\item for every nonterminal $A \in N$ there must 
exist a sibling sequence $w \in \text{sib}(t)$ such that
$\tilde{w}$ starts with $A$. For this sequence we have $e(w) = e(\tilde{w})+1$
(note that the right-hand side of $A$ does not belong to $\Sigma$,
hence $w$ starts with a tree of size at least 1).
\end{itemize}
Choose for every $A \in N$
a sibling sequence $w_A \in \text{sib}(t)$ such that
$\tilde{w}_A$ starts with $A$. Let 
$R = \text{sib}(t) \setminus  \{w_A \mid  A \in N\}$.
We get
\begin{eqnarray*}
|N|+\sum_{w \in \text{sib}(t)} e(\tilde{w}) & = & 
|N|+ \sum_{A \in N} e(\tilde{w}_A) + \sum_{w \in R} e(\tilde{w}) \\
& = & \sum_{A \in N} (e(\tilde{w}_A)+1) + \sum_{w \in R} e(\tilde{w}) \\
& \leq &  \sum_{A \in N} e(w_A) + \sum_{w \in R} e(w)\\
& = & \sum_{w \in \text{sib}(t)} e(w).
\end{eqnarray*}
This proves the theorem.
\qed
\end{proof}

\begin{theorem}\label{lemma:root_lemma}
For every $t \in \mathcal{T}(\Sigma)$ with $|t| \geq 2$, we have
$$
|\text{dag}(t)| \leq \frac{1}{2} |\text{hdag}(t)|^2.
$$
\end{theorem}

\begin{proof}
 Let $f_i(\alpha_{i,1}, \ldots, \alpha_{i,n_i})$ for $1 \leq i \leq k$ be the 
 right-hand sides of $\mathcal{G}_{\text{dag}(t)}^{\text{red}}$.
 W.l.o.g. assume that $1 \leq n_1 \leq n_2 \leq \cdots \leq n_k$.
 Every $\alpha_{i,j}$ is either from $\Sigma$ or a nonterminal.
 Moreover, all the trees $f_i(\alpha_{i,1}, \ldots, \alpha_{i,n_i})$ are
 pairwise different. We have
 $|\text{dag}(t)| = \sum_{i=1}^k n_i$.
 
 If $n_k = 1$, then $t$ is a linear chain. In this case, we get
 $$
 |\text{dag}(t)| = |t| \leq \frac{1}{2} |t|^2 =  \frac{1}{2}  |\text{hdag}(t)|^2
 $$
 since $|t| \geq 2$. Let us now assume that $n_k \geq 2$.
 Recall that we compute $\text{hdag}(t)$ by taking
 the minimal dag of the forest consisting
 of the binary encodings of the trees $f_i(\alpha_{i,1}, \ldots, \alpha_{i,n_i})$.
 The binary encoding of $f_i(\alpha_{i,1}, \ldots, \alpha_{i,n_i})$ has
 the form $f_i(t_i,\Box)$, where $t_i$ is a chain of $n_i-1$ many right pointers.
 Let $d$ be the minimal dag of the forest consisting of all chains
 $t_i$. Since all the trees $f_i(\alpha_{i,1}, \ldots, \alpha_{i,n_i})$
 are pairwise distinct, we have 
 $|\text{hdag}(t)| = k + |d|$.
 Since the chain $t_i$ consists of $n_i$ many nodes, we have
 $|d| \geq \max\{ n_i \mid 1 \leq i \leq k\}-1 =  n_k -1$.
 Hence, we have to show that
 $\sum_{i=1}^k n_i \leq \frac{1}{2} (k + n_k - 1)^2$. We have
 $$
 \sum_{i=1}^k n_i  \leq  k \cdot n_k 
   \leq  (k-1)  n_k + \frac{1}{2} n_k^2 =
   \frac{1}{2} (2 (k-1)  n_k + n_k^2) 
   \leq   \frac{1}{2} (k-1 + n_k)^2,
$$  
which concludes the proof. For the second inequality note that
$n_k \leq \frac{1}{2} n_k^2$, since $n_k \geq 2$.
\qed
\end{proof}
Consider the tree $s_n$ from Figure~\ref{fig:tn2}.
We have $|\text{dag}(s_n)| = |s_n| = n^2$
and $|\text{hdag}(s_n)| = |\text{bdag}(s_n)| = 3n-2$. 
Hence, we get
$$
|\text{dag}(s_n)| = n^2 > \frac{1}{9} (3n-2)^2 = \frac{1}{9}  |\text{hdag}(s_n)|^2 .
$$
This shows that
up to a constant factor, the bound in Theorem~\ref{lemma:root_lemma}
is sharp. The constant $1/9$ can be slightly improved:

\begin{theorem}\label{lemma:1/6}
There is a family of trees $(s_n)_{n \geq 1}$ such that
 $$|\mbox{dag}(s_n)| > \frac{1}{6}  |\text{hdag}(s_n)|^2.
 $$
\end{theorem}

\begin{proof}
We specify $s_n$ by the reduced $0$-SLT grammar $\mathcal{G}_{\text{dag}(s_n)}^{\text{red}}$. Let 
$\mathcal{G}_{\text{dag}(s_n)}^{\text{red}}$ contain the following 
productions for $0 \leq i \leq n$:
$$
A_i \to f(A_{i+1}, \ldots, A_n, \underbrace{a,\ldots,a}_{n \text{ many}}) .
$$
This is indeed the grammar obtained from the 
minimal dag for a tree $s_n$ (of size exponential in $n$).
We have 
$$
|\text{dag}(s_n)|  = \sum_{i=n}^{2n} i = n (n+1) +  \sum_{i=0}^{n} i = n (n+1) + \frac{n (n+1)}{2}   =   \frac{3 n (n+1)}{2}.
$$
The hybrid dag of $s_n$ consists of the child sequence $A_1 A_2 \cdots A_n a^n$ together with $n+1$ many left pointers
into this sequence. Hence, we have
$$
|\text{hdag}(s_n)|  = 2n-1 + n+1 = 3n .
$$
We obtain
$$
\frac{1}{6}  |\text{hdag}(s_n)|^2 = \frac{1}{6}  9 n^2 = \frac{3}{2} n^2 <   \frac{3 n (n+1)}{2} = |\text{dag}(s_n)| .
$$
This proves the theorem.
\qed
\end{proof}

Next let us bound $|\text{bdag}(t)|$ in terms of 
$|\text{hdag}(t)|$:

\begin{theorem}\label{lemma:half_lemma}
For every $t \in \mathcal{T}(\Sigma)$, we have
$$
|\text{bdag}(t)| + n \leq 2 |\text{hdag}(t)|,
$$
where $n$ is the number of non-leaf nodes of $\text{dag}(t)$.
\end{theorem}

\begin{proof}
We use the notations introduced before Theorem~\ref{hdag-kleiner-min}.
Note that $n = |N|$ is the number of nonterminals of the $0$-SLT grammar
$\mathcal{G}_{\text{dag}(t)}^{\text{red}}$. By Lemma~\ref{lemma-counting-edges-bdag}
we have  
$|\text{bdag}(t)|\ =\ \sum_{w \in \text{sib}(t)} e(w)$.
By Lemma~\ref{lemma-relating-sums} we have
$|\text{hdag}(t)|  \ = \ |N| + \sum_{w \in \text{sib}(t)} e(\tilde{w})$.
Hence, we have to show that
$$
|N| + \sum_{w \in \text{sib}(t)} e(\tilde{w}) \ \geq \ \frac{1}{2}  \sum_{w \in \text{sib}(t)} \!\! e(w) \ + \ \frac{1}{2} |N|.
$$
In order to prove this, we show the 
following for every sibling sequence $w \in \text{sib}(t)$:
Either $e(\tilde{w}) \geq \frac{1}{2} e(w)$ or
$e(\tilde{w}) = 0$ and $e(w) = 1$. In the latter case, the sibling sequence $w$ consists of a single tree $s$ of size 
at least one (i.e., $s$ does not consist of a single node), and $\tilde{w}$ consists of a single nonterminal $A \in N$.
So, let $w = t_1 \cdots t_n \in \text{sib}(t)$ and let $\tilde{w} = \alpha_1 \cdots \alpha_n$ with $\alpha_i \in \Sigma \cup N$.
We consider the following four cases:

\smallskip
\noindent
{\em Case 1.}  $n>1$ and $t_1 = \alpha_1 \in \Sigma$. We have $e(w) = e(\tilde{w})=1$.

\smallskip
\noindent
{\em Case 2.} $n > 1$ and $|t_1| \geq 1$. We have $e(w)=2$ and $e(\tilde{w})=1$.

\smallskip
\noindent
{\em Case 3.} $n=1$ and $t_1 = \alpha_1 \in \Sigma$. We have $e(w) = e(\tilde{w})=0$.

\smallskip
\noindent
{\em Case 4.} $n = 1$ and $|t_1| \geq 1$. We have $e(w)=1$, $e(\tilde{w})=0$, and 
$\tilde{w}$ consists of a single nonterminal $A \in N$.
\qed
\end{proof}
For the tree $t_m$ from Figure~\ref{fig:tn} we have
$|\text{bdag}(t_m)| = |t_m| = 2m$, $|\text{hdag}(s_m)| =  |\text{dag}(t_m) | = m+1$, and $n=|N|=2$.
Hence, Theorem~\ref{lemma:half_lemma} is optimal.

From Theorems~\ref{hdag-kleiner-min}, \ref{lemma:root_lemma}, and \ref{lemma:half_lemma} we immediately get:

\begin{corollary}
For every $t \in \mathcal{T}(\Sigma)$ with $|t| \geq 2$, we have
$$
\frac{1}{2} |\text{bdag}(t)| \le |\text{dag}(t)| \leq
\frac{1}{2} |\text{bdag}(t)|^2.
$$
\end{corollary}

\newcommand{\gf}{generating function}
\newcommand{\gfs}{generating functions}
\newcommand{\KN}{N}

\section{The average-case sizes of dag and bdag}

Let $m\ge 1$.
We will use the terminology {\em $m$-labeled tree}, instead of {\em $\{1, \ldots, m\}$-labeled tree}.
In this section,
we 
analyze the average sizes (both node size and edge size) of
the 
dags and bdags
of  $m$-labeled unranked trees
of size~$n$.
Currently, we are not able to make such an
analysis for the hdag. While 
Section~\ref{sec:exact}
%
provides exact expressions for the average sizes,
Section~\ref{subsection:asymptotic}
 deals with their asymptotic behavior.
The results are mostly an extension of~\cite{FlaSipStey1990}, where the authors
treat the average node size  of binary trees over a singleton alphabet
($m=1$). 
However, we give here complete proofs, whereas the proof was merely
sketched in \cite{FlaSipStey1990}.

Let $\mathcal{B}_m$
 denote the set of non-empty $m$-labeled
binary trees and let
$\mathcal{T}_m$ 
denote the set of  non-empty $m$-labeled unranked trees.
Here, ``non-empty'' means that our trees have at least one node.
%
For $\mathcal{U} \in \{ \mathcal{B}, \mathcal{T} \}$
and $n\ge 0$,
we define
$$
\mathcal{U}_{m,n} = \{ t \in \mathcal{U}_m\mid |t|=n \}.
$$ 
We seek expressions for the accumulated 
quantities 
$$
\KN _{m,n}^{\mathcal{U}} = \sum_{t \in \mathcal{U}_{m,n}} \| \text{dag}(t) \|
\quad \text{and}\quad
E_{m,n}^{\mathcal{U}} =  \sum_{t \in \mathcal{U}_{m,n}} |\text{dag}(t)|
$$
as well as for the average sizes
$$
\bar{\KN }_{m,n}^{\mathcal{U}} = \frac{\KN _{m,n}^{\mathcal{U}}}{|\mathcal{U}_{m,n}|}
\quad \text{and}\quad
\bar{E}_{m,n}^{\mathcal{U}} = \frac{E_{m,n}^{\mathcal{U}}}{|\mathcal{U}_{m,n}|}.
$$
Recall that the fcns-encoding yields a bijection between $m$-labeled unranked trees of edge size $n$ 
and $m$-labeled binary trees of edge size $n$, where the root only contains a left child. 
Therefore, the average node size (resp. edge size) of the bdag of $m$-labeled unranked trees of size $n$
is one plus the average node size (resp. edge size) of the minimal dag of $m$-labeled binary trees of size $n-1$.

One key tool 
used in this section is that of {\em generating functions}.
If $F_n$ is a sequence of numbers, then its (ordinary) 
 generating function is defined as
\[
 \mathbf F(z)= \sum_{n \ge 0} F_n z^n
\]
and $[z^n] \mathbf F(z)$ denotes the coefficient of $z^n$ in $\mathbf F(z)$ (i.e., $F_n$).
If for a 
set $\mathcal{S}$ a size function 
$f:\, \mathcal{S} \to \mathbb{N}$
is defined 
such that 
for every $n$ the set
$\mathcal{S}_n = \{ s \in \mathcal{S}\mid f(s)=n\}$ is finite, we can
associate to the set $\mathcal{S}$ the generating function 
\[
 \mathbf  S(z) = \sum_{n \ge 0} |\mathcal{S}_n| z^n,
\]
which is said to   \emm count the objects of $\mathcal S$ by their
size.,
Such sets $\mathcal S$ are 
sometimes called {\em combinatorial classes}~\cite[p.16]{AnalyticCombinatorics}.
When a class has a simple recursive structure, it is often possible to
obtain an explicit expression for $\mathbf S(z)$. This will be the case for our
two basic generating functions, $\mathbf B_m(z)$ and $\mathbf T_m(z)$, which count
respectively the trees of $\mathcal B_m$ and $\mathcal T_m$ by their
size (see Lemmas~\ref{lemma:gf_labeled_binary_trees} and~\ref{lemma:gf_labeled_unranked_trees}).

Let again $\mathcal{U} \in \{ \mathcal{B}, \mathcal{T} \}$.
For $u \in \mathcal{U}_m$ and $n\ge 0$,  define $C_{m,n}^{\mathcal{U}}(u)$ as the number of 
$\mathcal{U}_m$-trees 
of
size $n$ that {contain} $u$ as a subtree.
%
Let $ v \in \mathcal{U}_m$ be another tree such that $|v|=|u|$. For
every $n \ge 0$ there is 
a bijection between the set of trees of size $n$ that contain $u$ and the set of trees
of size $n$ that contain $v$
(it is obtained by replacing every occurrence of $u$ by a copy of $v$,
and vice versa).
Therefore 
$C_{m,n}^{\mathcal{U}}(u) =C_{m,n}^{\mathcal{U}}(v)$ 
and so we also write 
$C_{m,n}^{\mathcal{U}}(p)$ 
(with $p = |u|$) instead of
$C_{m,n}^{\mathcal{U}}(u)$. 
The corresponding generating function
is
$$
\mathbf C_{m,p}^{\mathcal{U}}(z) = \sum_{n \geq 0} C_{m,n}^{\mathcal{U}}(p)
z^n.
$$
This series will be determined in Lemma~\ref{lem:C-binary} for binary trees
and in Lemma~\ref{lemma:C^U_n} for unranked trees.
Let us now explain how the 
accumulated sizes
 $N^{\mathcal U}_{m,n}$ and
$E^{\mathcal U}_{m,n}$  (or, equivalently, the associated generating
functions) can be expressed in terms of these series.

%
Let $\text{sub}(t)$ denote the set of subtrees occurring in the tree $t$.
Since $\| \text{dag}(t) \| $ is 
 the number of different subtrees of 
$t$, we have ($\mathbbm{1}_{u\in  \text{sub}(t)}$ is $1$ if $u \in \text{sub}(t)$ and $0$ otherwise)
\begin{eqnarray*}   
 \KN _{m,n}^{\mathcal{U}} = \sum_{t \in \mathcal{U}_{m,n}}  \| \text{dag}(t) \|
 & = &  \sum_{t \in \mathcal{U}_{m,n}}  \sum_{u  \in \mathcal
   U_m}\mathbbm{1}_{u\in  \text{sub}(t)}\\
 & = &  \sum_{u \in \mathcal{U}_m} C_{m,n}^{\mathcal{U}}(u)
 = \sum_{p \geq 0} |\mathcal{U}_{m,p}| \, C_{m,n}^{\mathcal{U}}(p).
\end{eqnarray*}
Hence
the corresponding generating function is 
\begin{equation}\label{formula-K(z)^U-general}
\mathbf \KN _m^{\mathcal{U}} (z)
= \sum_{n \geq 0} \KN _{m,n}^{\mathcal{U}} z^n =    \sum_{p \geq 0} |\mathcal{U}_{m,p}| \, \mathbf C_{m,p}^{\mathcal{U}}(z).
\end{equation}
We now want to obtain  an expression for the accumulated number of edges
$E_{m,n}^{\mathcal U}$ and the associated generating function.
Let us denote by 
${U}_{m,n}^{(d)}$ (with $U=B$ or $T$)
the number of 
 trees from $\mathcal U_{m,n}$ that have root degree $d$ (i.e., the root has $d$ children).
Then, in the same spirit as for the number of nodes, we get for the
number of edges:  
\begin{eqnarray*}
  E_{m,n}^{\mathcal{U}} = \sum_{t \in \mathcal{U}_{m,n}} |\text{dag}(t)| 
   & = & \sum_{t \in \mathcal{U}_{m,n}}\sum_{u\in \text{sub}(t)}
   \deg(\text{root}(u) )\\
& = & \sum_{u \in\mathcal U_m} \deg(\text{root}(u)) C^{\mathcal U}_{m,n}(u)= \sum_{p,d \geq 0} d\, {U}_{m,p}^{(d)} C_{m,n}^{\mathcal{U}}(p).
\end{eqnarray*}
 The associated generating function is
\begin{equation}\label{formula-E(z)^U-general}
 \mathbf E_m^{\mathcal{U}} (z)
   = \sum_{n \geq 0} E_{m,n}^{\mathcal{U}} z^n =  \sum_{p,d \geq 1} d\, {U}_{m,p}^{(d)} \mathbf C_{m,p}^{\mathcal{U}}(z).   
\end{equation}
Indeed, we can ignore trees of size $p=0$ (or, equivalently, root degree $d=0$).
%
\subsection{Exact counting}\label{sec:exact}
%
In this section, we determine explicit expressions for the
\gfs\  $\mathbf N^{\mathcal{U}}_m (z)$ and  $\mathbf E_m^{\mathcal{U}} (z)$, whose
coefficients record the accumulated number of nodes (resp. edges) in the
dag of $m$-labeled 
trees of size $n$.
We start with binary trees (that is, $\mathcal U=\mathcal B$).

\subsubsection{Binary trees}

\begin{lemma}\label{lemma:gf_labeled_binary_trees}
 The generating function $\mathbf B_m(z)$ of $m$-labeled binary trees, 
counted by their edge number, 
is
 \begin{align}
  \mathbf B_m(z) &= \frac{1-2mz-\sqrt{1-4mz}}{2mz^2}. \label{eq:B_m1}
      \end{align}
Equivalently, the number of $m$-labeled binary trees of size $p$ is
\beq\label{eq:B_m2}
B_{m,p}= \frac 1{p+2} {2p+2\choose p+1} m^{p+1}.
\eeq
Of course the case  $m=1$  recovers the (shifted) Catalan numbers.
\end{lemma}
%
\begin{proof}
The proof of Lemma~\ref{lemma:gf_labeled_binary_trees} follows a
general principle called the {\em symbolic method} in~\cite[Chapter 1]{AnalyticCombinatorics}:
If a combinatorial class $\mathcal{H}$ is built by disjoint union
from two combinatorial classes $\mathcal{F}$ and $\mathcal{G}$ with the respective generating
functions $\mathbf F(z)$ and $\mathbf G(z)$, 
then the generating function $\mathbf H(z)$ of the combinatorial class
$\mathcal{H}$ is $\mathbf H(z) = \mathbf F(z) + \mathbf G(z)$. Similarily, if $\mathcal{H}$ is build via
Cartesian
 product from the classes $\mathcal{F}$ and $\mathcal G$, 
 then $\mathbf H(z) = \mathbf F(z)\cdot \mathbf G(z)$.
 
An $m$-labeled binary tree is either a single node 
with a label from $\{1, \ldots, m\}$,
or a root node with a single 
subtree 
(left or right)
or a root node with two
subtrees. 
 The above principles give for the \gf\ $ \mathbf B_m(z)$ the equation
\[
 \mathbf B_{m}(z) = m + 2mz\mathbf B_{m}(z) + m \left( z\mathbf B_{m}(z) \right)^2.
\]
Solving this equation for $\mathbf B_{m}(z)$ proves equation~\eqref{eq:B_m1}
(taking the other root for $\mathbf B_m(z)$ would give a series with negative
powers of $z$). 
Equation~\eqref{eq:B_m2} follows from~\eqref{eq:B_m1} by 
a Taylor expansion.
\qed
\end{proof}

\begin{lemma}\label{lem:C-binary}
 The generating function $\mathbf C_{m,u}^{\mathcal{B}}(z)$ of $m$-labeled binary trees
 that contain a given tree $u \in \mathcal{B}_m$ of size $p$ is
$$
 \mathbf C_{m,p}^{\mathcal{B}}(z) =\frac{1}{2mz^2}
 \left(\sqrt{1-4mz+4mz^{p+2}} - \sqrt{1-4mz} \right).
$$
\end{lemma}
\begin{proof}
We first determine the \gf\  $\mathbf A_{m,p}^{\mathcal{B}}(z)$ counting $m$-labeled binary
trees that do  \emm not, contain (or {\em avoid}) $u$. A non-empty
binary tree $t$ that avoids 
$u$ is either 
reduced to a single node,
or 
a root node with a $u$-avoiding tree attached (this may
be the left or the right child),
or a root node to which two $u$-avoiding trees are attached. 
However, we must still exclude the tree
$t=u$, which is included in the above recursive description.  We thus get the following equation:
\[
 \mathbf A_{m,p}^{\mathcal{B}}(z) = 
 m + 2mz \mathbf A_{m,p}^{\mathcal{B}}(z) + m \left(z \mathbf A_{m,p}^{\mathcal{B}}(z) \right)^2 -z^p,
\]
which
yields
\[
 \mathbf A_{m,p}^{\mathcal{B}}(z) = \frac{1-2mz-\sqrt{1-4mz + 4mz^{p+2}}}{2mz^2}.
\]
Using $\mathbf C_{m,p}^{\mathcal{B}}(z) = \mathbf B_m(z) - \mathbf A_{m,p}^{\mathcal{B}}(z)$,
this proves 
the lemma.
\qed
\end{proof}
We now obtain expressions for the \gfs\   
$\mathbf \KN_m^{\mathcal{B}}(z)$ and   $\mathbf E_m^{\mathcal{B}}(z)$ given
by~\eqref{formula-K(z)^U-general} and~\eqref{formula-E(z)^U-general}.
\begin{theorem}\label{lemma:binary_K}
 The generating function of the accumulated number of nodes of minimal dags of
$m$-labeled
 binary trees is
 $$
 \mathbf  \KN _m^{\mathcal{B}}(z) = \frac{1}{2mz^2} \sum_{p \ge 0} B_{m,p}
  \left( \sqrt{1-4mz+4mz^{p+2}} -\sqrt{1-4mz}
  \right),
$$
where the numbers $B_{m,p}$ are given by~\eqref{eq:B_m2}.

%
  The generating function of the accumulated number of edges of dags of
$m$-labeled  binary trees is
$$
  \mathbf E_m^{\mathcal{B}}(z) = \frac{3}{2mz^2} \sum_
{p \ge 1} 
\frac p{2p+1} B_{m,p}\left( \sqrt{1-4mz+4mz^{p+2}} -\sqrt{1-4mz} \right).
$$
 \end{theorem}
Equation~(3) in~\cite{FlaSipStey1990} can be obtained from the
above expression for $\mathbf \KN _m^{\mathcal{B}}(z)$  by setting
$m=1$ and by shifting the index
(since the size is defined as  the number of nodes in~\cite{FlaSipStey1990}).
%

\begin{proof}
 The expression for $\mathbf \KN^{\mathcal B}_{m}(z)$ follows directly
 from~\eqref{formula-K(z)^U-general} and Lemma~\ref{lem:C-binary}.
To express the series $\mathbf E^{\mathcal B}_{m}(z)$, we first need to
determine (according to~\eqref{formula-E(z)^U-general}) the number
${B}_{m,p}^{(d)}$  of
$m$-labeled binary trees of size $p \geq 1$ with root degree~$d$. Note that $d$
can only be  1 or 2.  Clearly, 
  %
${B}_{m,p}^{(1)}= 2mB_{m,p-1}$,
 and thus  ${B}_{m,p}^{(2)}= B_{m,p}-2mB_{m,p-1}$.
%
Hence, for $p\ge 1$,
\begin{eqnarray*}
\sum_{d\ge 1} d \cdot {B}_{m,p}^{(d)} & = & 2mB_{m,p-1}+ 2(B_{m,p}-2mB_{m,p-1})\\
& = & 2(B_{m,p}-mB_{m,p-1}) \\
& = & \frac{3p}{2p+1} B_{m,p},
\end{eqnarray*}
where the last equation follows from \eqref{eq:B_m2}.
The expression for $\mathbf E^{\mathcal B}_{m}(z)$ now follows,
using~\eqref{formula-E(z)^U-general} and Lemma~\ref{lem:C-binary}.
    \qed
\end{proof}

\subsubsection{Unranked trees}
%
\begin{lemma}\label{lemma:gf_labeled_unranked_trees}
 The generating function $\mathbf T_m(z)$ of $m$-labeled unranked trees is
\begin{align}
 \mathbf{T}_m(z) &= \frac{1-\sqrt{1-4mz}}{2z}.  \label{eq:T_m1}                
\end{align}
Equivalently, the number of $m$-labeled unranked trees of size $p$ is
\beq
 T_{m,p}= \frac{1}{p+1} \binom{2p}{p} 
m^{p+1}.
\label{eq:T_m2}
\eeq
Again, we obtain the Catalan numbers when $m=1$.
\end{lemma}
\begin{proof}
 An 
$m$-labeled unranked tree is a 
root
 node
to which a sequence of $m$-labeled unranked trees is attached.
 We can now use another construction
from~\cite[Chapter 1]{AnalyticCombinatorics}:
If $\mathcal{G}$ is a combinatorial class that does not contain an
element of size $0$, and 
the class $\mathcal{F}$ is defined as
\[
 \mathcal{F} = \{\epsilon \} + \mathcal{G} + (\mathcal{G} \times \mathcal{G}) + 
 (\mathcal{G} \times \mathcal{G} \times \mathcal{G}) + \cdots,
\]
then the generating function 
of $\mathcal{F}$ is
\[
 \mathbf F(z) = \frac{1}{1-\mathbf G(z)}
\]
where $\mathbf G(z)$ is the generating function 
of $\mathcal{G}$.

In our case, $\mathbf G(z)=z \mathbf T_m(z)$ counts trees with root degree 1 and root
label~1,
and we thus obtain
\begin{equation*}
 \mathbf{T}_m(z) = \frac{m}{1- z \mathbf{T}_m(z)}.
\end{equation*}
Solving this for $\mathbf{T}_m(z)$ yields~\eqref{eq:T_m1}.
We then obtain equation~\eqref{eq:T_m2} by a Taylor expansion.
\qed
\end{proof}

%
\begin{lemma}\label{lemma:C^U_n}
 The generating function of 
$m$-labeled unranked trees 
that contain a given tree $u$ of size $p$ is
$$
  \mathbf C_{m,p}^{\mathcal{T}}(z) = \frac{z^{p+1}+ \sqrt{1- 4mz +
      2z^{p+1} + z^{2p+2} } - \sqrt{1-4mz} }{2z}. 
$$
\end{lemma}
\begin{proof}
We first determine the \gf\  $\mathbf A_{m,p}^{\mathcal{T}}(z)$ counting
$m$-labeled unranked 
trees that do \emm not, contain (or avoid) $u$. A 
tree that avoids 
 $u$
is a root node to which a sequence of $u$-avoiding trees is attached. 
As in the binary case, we still need to subtract $z^p$ to avoid
counting 
 $u$ itself. This gives
\begin{align*}
 \mathbf A_{m,p}^{\mathcal{T}}(z) &= \frac{m}{1 - z\mathbf A_{m,p}^{\mathcal{T}}(z)} -z^p,
\end{align*}
which can be solved for $\mathbf A^{\mathcal{T}}_{m,p}(z)$:
\begin{align*}
 \mathbf A_{m,p}^{\mathcal{T}}(z) &= \frac{1}{2z} \left( 1 - z^{p+1} - \sqrt{1- 4mz + 2z^{p+1} + z^{2p+2}}\right).
\end{align*}
Using $\mathbf C_{m,p}^{\mathcal{T}}(z) = \mathbf T_m(z) - \mathbf A_{m,p}^{\mathcal{T}}(z)$,
this proves the lemma.
\qed
\end{proof}
%
\begin{proposition}\label{prop:NE-unranked}
  The generating function of the accumulated node size of minimal dags of 
$m$-labeled unranked trees is
$$
  \mathbf  \KN _m^{\mathcal{T}}(z) = \frac{1}{2z} \sum_{p \ge 0} T_{m,p} 
\left(z^{p+1}+ \sqrt{1-4mz+2z^{p+1} +z^{2p+2}} - \sqrt{1-4mz}\right),
$$
where the numbers $T_{m,p}$ are given by~\eqref{eq:T_m2}.

 The generating function of the accumulated edge size of minimal dags of 
$m$-labeled unranked trees is
 \begin{align*}
    \mathbf  E_m^{\mathcal{T}}(z) &= \frac{3}{2z} \sum_{p\ge 0} \frac{p
       T_{m,p}}{p+2} \left( z^{p+1}+ \sqrt{1-4mz+2z^{p+1} +z^{2p+2}} - \sqrt{1-4mz}\right).
  \end{align*}
\end{proposition}
\begin{proof}
The expression of $\mathbf \KN^{\mathcal T}_m(z)$ follows directly
from~\eqref{formula-K(z)^U-general} and Lemma~\ref{lemma:C^U_n}.

To express the series $\mathbf E^{\mathcal T}_m(z)$, we first need to
determine (according to~\eqref{formula-E(z)^U-general}) the number
$T^{(d)}_{m,p}$ of $m$-labeled unranked trees of size $p$ and root
degree $d$, or, more precisely, the sum
$$
\sum_{d\ge 1} d \cdot T^{(d)}_{m,p}
$$
for any $p\ge 1$. This is done in~\cite[Corollary~4.1]{DBLP:journals/dm/DershowitzZ80} in the case
$m=1$. It
suffices to multiply by $m^{p+1}$ to obtain the general case:
$$
\sum_{d\ge 1} d \cdot T^{(d)}_{m,p} 
%
= \frac{3pT_{m,p} }{p+2}.
$$
Combining~\eqref{formula-E(z)^U-general} and Lemma~\ref{lemma:C^U_n}
now gives the expression of $\mathbf E^{\mathcal T}_m(z)$.
\qed
\end{proof}
\subsection{Asymptotic results}\label{subsection:asymptotic}
In this section we 
state asymptotic results for the average 
node and edge sizes of the dag
of $m$-labeled binary trees, and of $m$-labeled unranked trees.
The proofs are rather involved and assume some knowledge in analytic combinatorics \cite{AnalyticCombinatorics}.
Therefore, the proofs are  are given in
the Appendix.

\subsubsection{Binary trees}
\begin{theorem}\label{theorem:asymptotic_binary_node}
The average number of nodes in the minimal dag of an $m$-labeled binary tree
of size $n$  satisfies
\begin{equation}
   \bar{\KN }_{m,n}^{\mathcal{B}} 
   =  
2 
\kappa_m\frac{n}{\sqrt{\ln n}} \left( 1+O \left( \frac{1}{\ln n} \right) \right)
 \quad \text{with} \quad 
\kappa_m = 
\sqrt{\frac{\ln (4m)}{\pi }}.\label{eq:asymptotic_binary_node}
\end{equation}
\end{theorem}
The proof is an application of the \textit{singularity
analysis} of Flajolet and Odlyzko, described
in~\cite[Ch.~VI]{AnalyticCombinatorics}.  One first determines
the singular behavior of the series $\mathbf \KN_m ^{\mathcal{B}}(z)$ 
given by Theorem~\ref{lemma:binary_K} in the neighborhood of its
\emm dominant, singularities 
(that is, singularities of minimal modulus). 
%
\begin{theorem}\label{lemme1-fl}
The generating function $\mathbf \KN _m^{\mathcal{B}}(z)$
 is analytic in the domain $D$ 
defined by $|z|< \frac{1}{2m}$ and $z \notin
[ \frac{1}{4m},\frac{1}{2m} ]$. As $z$ tends to $ \frac{1}{4m}$ in $D$,
one has
\begin{equation*}
\mathbf  \KN_m ^{\mathcal{B}}(z) = 
\frac{8 \, m \, \kappa_m } {\sqrt{(1-4m z ) \ln ((1-4mz)^{-1})}} + O \left(\frac{1}{\sqrt{(1-4mz)\ln^3((1-4mz)^{-1})}} \right),
\end{equation*}
where $\kappa_m$ is defined as in Theorem~\ref{theorem:asymptotic_binary_node}.
\end{theorem}
Granted this proposition,
one can use 
the Transfer Theorem~VI.4 of~\cite[p.~393]{AnalyticCombinatorics},
combined with the estimates of
the coefficients of elementary series 
(see~\cite[Thm.~VI.2, p.~385]{AnalyticCombinatorics}) 
to obtain the asymptotic behavior of the accumulated node size of minimal dags 
of $m$-labeled binary trees of size $n$:
$$
\KN _{m,n}^{\mathcal{B}} = [z^n]\mathbf \KN_m ^{\mathcal{B}}(z)
= \frac{2\kappa_m}{\sqrt{\pi}}
\frac {4^{n+1} m^{n+1}}{\sqrt{n \ln n}}\left(1+ O\left(\frac{1}{\ln n}\right)\right).
$$
Since the numbers $B_{m,n}$, given by~\eqref{eq:B_m2}, satisfy
$$
B_{m,n} = \frac{4^{n+1} m^{n+1}}{\sqrt \pi n^{3/2}}\left(1+ O\left(\frac{1}{n}\right)\right),
$$
this gives Theorem~\ref{theorem:asymptotic_binary_node}.
The proof of Theorem~\ref{lemme1-fl} can be found in the Appendix,
Section~\ref{proof:fl}
(for $m=1$) and Section~\ref{sec:binary-m} (for general values of $m$).

\smallskip
For the  the edge size, one obtains in a similar fashion the following
result. 
\begin{theorem}\label{lemma:asymptotic_binary_edge}
The average number of edges in the minimal dag of an $m$-labeled binary tree
of size $n$  satisfies
\begin{equation}
   \bar{E}_{m,n}^{\mathcal{B}} 
   = {3  \kappa_m} \frac{n}{\sqrt{\ln n}} \left( 1+O\left( \frac{1}{\ln n} \right) \right)
\end{equation}
with $\kappa_m$ as in Theorem~\ref{theorem:asymptotic_binary_node}.
\end{theorem}
The proof 
 is a simple adaptation of the proof of
 Theorem~\ref{theorem:asymptotic_binary_node} and can be found in
Section~\ref{sec:binary-e} of
 the Appendix.
Note the factor $3/2$ between the node and edge sizes, which could
be predicted by comparing the  expressions of $\mathbf \KN^{\mathcal B}_m(z)$
and $\mathbf E^{\mathcal B}_m(z)$
in Theorem~\ref{lemma:binary_K}.

\subsubsection{Unranked trees}
      
\begin{theorem}\label{theorem:asymtotic_node_size_unranked}
The average number of nodes in the minmal dag of an $m$-labeled unranked tree
of size $n$  satisfies
  \begin{equation}
    \bar{\KN }^{\mathcal{T}}_{m,n} = {\kappa_m } \frac{n}{\sqrt{\ln n}} \left( 1+O\left( \frac{1}{\ln n} \right) \right),
  \end{equation}
with $\kappa_m$ as in Theorem~\ref{theorem:asymptotic_binary_node}.
\end{theorem}
Thus the average
 node size of  compressed unranked trees is about half
the node size of 
 compressed binary trees of the same size. Note  that the same ratio
holds between the heights of these trees~\cite{debruijn,flajolet-trees,marckert}.

The proof of Theorem~\ref{theorem:asymtotic_node_size_unranked}
 is very similar to the proof of
Theorem~\ref{theorem:asymptotic_binary_node}.  The required changes
are described in Section~\ref{sec:unranked-n}.

\begin{theorem}\label{thm:unranked-e}
The average number of edges in the minimal dag of an $m$-labeled unranked tree
of size $n$ satisfies 
  \begin{equation}
  \bar{E}_n^{\mathcal{T} }= 3 \kappa_m \frac{n}{\sqrt{\ln n}} \left( 1+O\left( \frac{1}{\ln n} \right) \right),
  \end{equation}
with $\kappa_m$ as in Theorem~\ref{theorem:asymptotic_binary_node}.
\end{theorem}
In other words, asymptotically the edge size of compressed binary trees is equal to the edge size of compressed unranked trees.
The proof  of Theorem~\ref{thm:unranked-e}
is given in Section~\ref{sec:unranked-e}.

Table~\ref{table:asymtpics_overview} contains an overview of the results of this section.

\begin{table*}[t]
\centering
\begin{tabular}{|l|c|c|}
\hline
  & $\mathcal{B}_m$ & $\mathcal{T}_m$ \\ 
\hline
&&\\
$\displaystyle \bar{\KN }_{m,n}$
&  $\displaystyle 2 \kappa_m \frac{n}{\sqrt{\ln  n }} \left( 1 + O\left( \frac{1}{ \ln  n }\right) \right)$ 
&  $\displaystyle \kappa_m \frac{n}{\sqrt{\ln n}} \left( 1 + O\left( \frac{1}{ \ln
  n}\right) \right)$ \\ 
&&\\
$\displaystyle \bar{E}_{m,n}$
&  $\displaystyle3 \kappa_m \frac{n}{\sqrt{\ln  n }} \left( 1 + O\left( \frac{1}{ \ln  n }\right) \right)$
&  $\displaystyle3 \kappa_m \frac{n}{\sqrt{\ln n}} \left( 1 + O\left( \frac{1}{ \ln
  n}\right) \right)$\\ 
&&\\ \hline
\end{tabular}
\caption{\label{table:asymtpics_overview} Overview over the different asymptotics. Recall that 
$\kappa_m =  \sqrt{\frac{\ln 4m}{\pi }}$.}
\end{table*}

\section{DAG and string compression}\label{sec:dag_plus_string}

As for the hdag, consider the forest $\text{fcns}(t_1), \ldots, \text{fcns}(t_n)$ 
of the binary encodings of the right-hand sides $t_1, \ldots, t_n$ 
of the reduced
$0$-SLT grammar $\mathcal{G}_{\text{dag}(t)}^{\text{red}}$
for an unranked tree $t$. 
In the construction of the hdag we build the minimal dag of this forest. 
Therefore we only share repeated suffixes of child sequences, i.e., 
``right branching'' trees in the binary encodings. Such trees can in fact be considered
as \emph{strings}. We now want to generalize the sharing of 
suffixes. 
Instead of only sharing suffixes of child sequences, we now apply
an arbitrary grammar-based string compressor to (a concatenation of)
the child sequences. Such a compressor infers a small straight-line context-free
grammar for the given string.

Formally, a \emph{straight-line context-free
string grammar, SL grammar} for short, is a triple $G = (N, \Sigma,\rho)$, where
\begin{itemize}
\item $N$ is a finite set of nonterminals
\item $\Sigma$ is a finite set of terminal symbols
\item $\rho :  N \to (N \cup \Sigma)^*$ is a mapping such
that the binary relation $\{ (X,Y) \mid X,Y \in N,  \rho(X) \in (N \cup \Sigma)^* Y (N \cup \Sigma)^* \}$
is acyclic.
\end{itemize}
We do not need a start nonterminal for our purpose.
From every word $u \in (N \cup \Sigma)^*$ we can derive exactly one terminal string 
$\text{eval}_G(u)$ using the mapping $\rho$. Formally, we extend $\rho$ to a morphism
$\rho : (N \cup \Sigma)^* \to (N \cup \Sigma)^*$ by $\rho(a) = a$ for $a \in \Sigma$. Due
to the above acyclicity condition, for every $u \in  (N \cup \Sigma)^*$, there exists
an $n \geq 1$ with $\rho^n(u) \in \Sigma^*$, and $\text{eval}_G(u)$ is this string.
We define the size of $G$ as $\sum_{X \in N} |\rho(X)|$.
As for SLTs we also write $X \to u$ if $\rho(X) = u$.

An \emph{SL grammar-compressed $\Sigma$-labeled dag} is a tuple $D=(V, \gamma, \lambda, G)$ such that
the following holds:
\begin{itemize}
\item $G = (N, V, \rho)$ is an SL grammar with terminal alphabet $V$
\item $\gamma : V \to (V \cup N)^*$
\item $\lambda : V \to \Sigma$ and
\item the triple $d=(V, \gamma', \lambda)$ with $\gamma'(v) = \text{eval}_G(\gamma(v)) \in V^*$ is an
ordered $\Sigma$-labeled rooted dag.
\end{itemize}
We define $\text{eval}(D)=\text{eval}(d)$.
We define the size $|D|$ of $D$ as $|G| + \sum_{v \in V} |\gamma(v)|$.
We say that $D$ is {\em minimal} if $d$ is the minimal dag for $\text{eval}(D)$.
Note that there are many minimal SL grammar-compressed dags for a given tree,
since we do not make any restrictions on the SL grammar part $G$. In particular,
$G$ does not have to be size minimal.

\begin{example} \label{ex-grammar-compressed-dag}
Here is an example of an SL grammar-compressed $\Sigma$-labeled dag $D=(V, \gamma, \lambda, G)$ with
$\Sigma = \{a,b,c,f,g,h\}$ and
$$V = \{A_1, A_2, A_3, A_4,A,B,C\}.
$$
The mappings $\gamma$ and $\lambda$ are shown below in the left and middle column in form of a $0$-SLT grammar.
For instance, $A_1 \to f(A,D,A_4,D,C)$  stands for $\lambda(A_1) = f$ and $\gamma(A_1) = ADA_4DC$.
The SL grammar $G$ is shown in the right column;
it contains the nonterminals $D$ and $E$.
\begin{alignat*}{3}
A_1 & \to f(A,D,A_4,D,C) &   \qquad  A & \to a  & \qquad     D & \to A_2 A_3 \\
A_2 & \to g(E,A)              &   \qquad    B & \to b &   \qquad    E & \to AA \\
A_3 & \to h(E,B)              &   \qquad    C & \to c &   \qquad       & \\
A_4 & \to f(D)                   &   \qquad      &
\end{alignat*}
The size of this SL grammar-compressed dag
is $14$ and 
it represents the dag $d$ with the following $0$-SLT grammar $\mathcal{G}_d$:
\begin{alignat*}{2}
A_1 &\to f(A,A_2,A_3,A_4 ,A_2,A_3 ,C) & \qquad A & \to  a  \\
A_2 &\to g(A,A,A) & \qquad B & \to  b  \\
A_3 &\to h(A,A,B) & \qquad C & \to  c \\
A_4 &\to f(A_2 ,A_3).  & & 
\end{alignat*}
Also note that $D$ is minimal.
\end{example}
By the following theorem, a given SL grammar-compressed dag for a tree $t$ can be efficiently
transformed into a $1$-SLT grammar that produces the binary encoding of $t$.

\begin{theorem} \label{thm-construct-1-STL}
An SL grammar-compressed $\Sigma$-labeled dag $D = (V, \gamma, \lambda, G)$ can be transformed in time
$O(|D|)$ into a $1$-SLT grammar $G_1$ such that $\text{eval}(G_1) = \text{fcns}(\text{eval}(D))$ and $|G_1| \leq |D|+2(|V|+|N|)$.
\end{theorem}

\begin{proof}
Let $G = (N,V,\rho)$.
Let $\hat{V} = \{ \hat{v} \mid v \in V\}$,  $V' = \{ v' \mid v \in V\}$,
$\hat{N} = \{ \hat{X} \mid X \in N\}$, and $N' = \{ X' \mid X \in N\}$ 
be disjoint copies of the sets $V$ and $N$, respectively. The set of nonterminals of the 
$1$-SLT grammar $G_1$ is $N \cup N' \cup \hat{N} \cup V \cup \hat{V} \cup V'$.
Nonterminals in $N \cup V \cup V'$ have rank $0$ and nonterminals in $\hat{N} \cup  N' \cup \hat{V}$
have rank $1$. 
The idea is that $\hat{\alpha}$ (for $\alpha \in N \cup V$) represents a copy of $\alpha$ 
that appears at positions  in the fcns encoding having exactly
one child (a right child), whereas
the original $\alpha$ will only appear in leaf positions. 
This distinction is necessary since in an SLT grammar every nonterminal has a
fixed rank. Nonterminals in $N' \cup V'$ are used in order to keep $G_1$ small.

The right-hand side mapping of $G_1$ is defined as follows:
For every $v \in V$ with $\lambda(v) = f$ and $\gamma(v) = \alpha_1 \cdots \alpha_k$
($k \geq 0$, $\alpha_1, \ldots, \alpha_k \in V \cup N$) we set:
\begin{itemize}
\item If $k=0$, then $v \to f$ and $\hat{v}(y) \to f(\Box,y)$; the nonterminal  $v'$ is not needed in this case.
\item If $k \geq 1$, then  $v \to  f(v', \Box)$, $\hat{v}(y) \to  f(v', y)$, and 
$v' \to \hat{\alpha}_1( \cdots \hat{\alpha}_{k-1}( \alpha_k) \cdots )$.
\end{itemize}
Note that  the total size of these productions is at most $k+2$ (recall that we do not count edges
to $\Box$-labeled nodes). Removing the 
nonterminal $v'$ in the case
$k \geq 1$ would result in a total size of $2k+1$.

For the every $X \in N$ with $\rho(X) = \beta_1 \cdots \beta_m$ ($\beta_1, \ldots, \beta_m \in N \cup V$, $m \geq 2$ with loss of generality) 
we set $X \to X'(\beta_m)$,  $\hat{X}(y) \to X'(\hat{\beta}_m(y))$, and 
$X'(y) \to  \hat{\beta}_1( \cdots \hat{\beta}_{m-1}(y) \cdots)$. 
These rules have total size $m+2$. 
Hence, the size of the resulting $1$-SLT
grammar $G_1$ is $|D|+2(|V|+|N|)$ and the time needed to construct
 it is clearly bounded by $O(|G_1|) = O(|D|)$. It is easy to see that $G_1$ produces $\text{fcns}(\text{eval}(D))$.
\qed
\end{proof}
Theorem~\ref{thm-construct-1-STL} implies 
that results for 1-SLT grammars carry over to SL grammar-compressed dags.
For instance, 
finite tree  automata~\cite{DBLP:journals/tcs/LohreyM06}
(with sibling constraints \cite{DBLP:journals/jcss/LohreyMS12})
and tree-walking automata \cite{DBLP:journals/jcss/LohreyMS12} can be evaluated in polynomial time
over 1-SLT grammars and hence over SL grammar-compressed dags.

To construct an SL grammar-compressed dag for a tree $t$, we first construct 
$\text{dag}(t)$ in linear time. Then we apply a grammar-based string compressor
(e.g., RePair~\cite{DBLP:conf/dcc/LarssonM99} or
Sequitur~\cite{DBLP:journals/jair/Nevill-ManningW97})
to the child sequences of the dag.   In this second phase 
we want to derive a small SL grammar for a set of strings and not 
a single string. To do this, we concatenate all child sequences of the dag separated
by unique symbols. For instance, for the dag at the end of Example~\ref{ex-grammar-compressed-dag}
we obtain the string
$$
A A_2 A_3 A_4 A_2 A_3 C \$_1 A A A  \$_2 A A B \$_3 A_2 A_3 .
$$
An application of RePair to this string yields the grammar
$$
S \to A D A_4 D C \$_1 E A \$_2 E B \$_3 D, \qquad D \to A_2 A_3, \qquad E \to A A  .
$$
Then, the right-hand side $A D A_4 D C \$_1 E A \$_2 E B \$_3 D$ contains the right-hand sides 
of the $\gamma$-mapping of the SL grammar-compressed dag, whereas the two remaining productions
$D \to A_2 A_3$ and $E \to AA$ make up the SL grammar part.

The following example shows 
that our construction may compress $\text{dag}(t)$ 
exponentially. 

\begin{example}
Consider the tree $f(a,a,\dots,a)$ with 
$2^n$ many $a$-leaves. Its dag has $2^n$ many
edges.
We apply a grammar-based string  compressor to the string
$a^{2^n}$. The string compressor may produce the string grammar
\begin{eqnarray*}
S' &\to& A_1 A_1 \\
A_i&\to& A_{i+1}A_{i+1} \text{ for }1\leq i\leq n-2\\
A_{n-1}&\to& aa
\end{eqnarray*}
of size $2n$. Actually, RePair would produce such a grammar.
The construction from the proof of Theorem~\ref{thm-construct-1-STL}
yields the following 1-SLT grammar, where we eliminate
productions that do not reduce the total grammar size:
\begin{eqnarray*}
S &\to & f(\hat{A}_1(A_1),\Box)  \\
A_i &\to & \hat{A}_{i+1}(A_{i+1})  \text{ for }1\leq i\leq n-2 \\
\hat{A}_i(y) & \to & \hat{A}_{i+1}(\hat{A}_{i+1}(y)) \text{ for }1\leq i\leq n-2 \\
A_{n-1} & \to & a(\Box,a) \\
\hat{A}_{n-1}(y) & \to & a(\Box,a(\Box,y)) 
\end{eqnarray*}
The total size of this grammar is $3n-1$.
Hence we obtain a 1-SLT grammar for the fcns encoding of 
$f(a,a,\dots,a)$ of size $O(n)$.
\end{example}
Both, $\text{hdag}(t)$ and $\text{rhdag}(t)$ can be seen as particular 
minimal SL grammar-compressed dags. For instance, $\text{hdag}(t)$ can be seen
as a minimal SL grammar-compressed dag $D = (V, \gamma, \lambda, G)$, where
the SL grammar $G = (N, V, \rho)$ is \emph{right regular}, i.e., for every nonterminal
$X \in N$ we have $\rho(X) \in V^* N \cup V^+$, and similarly, for every
$v \in V$ we have $\gamma(v) \in V^* N \cup V^*$.
 When transforming such an SL grammar
compressed dag into a $1$-SLT grammar following the proof of Theorem~\ref{thm-construct-1-STL},
we do not need the copy sets $\hat{N}$ and $N'$, because nonterminals from $N$ always produce
suffixes of child sequences in the dag.
This implies the following: 

\begin{theorem} \label{thm-construct-1-STL-right-reg}
An hdag that is represented as an
SL grammar-compressed $\Sigma$-labeled dag $D = (V, \gamma, \lambda, G)$
can be transformed in time
$O(|D|)$ into a $1$-SLT grammar $G_1$ such that $\text{eval}(G_1) = \text{fcns}(\text{eval}(D))$ and $|G_1| \leq |D|+2|V|$.
\end{theorem}

\section{Subtree equality check}

In the previous sections we have discussed five different formalisms 
for the compact representation of unranked trees:
\begin{enumerate}[(1)]
\item dag
\item binary dag
\item hybrid dag
\item SL grammar-compressed dag
\item SLT grammars (e.g. produced by BPLEX or TreeRePair)
\end{enumerate}
As mentioned in Section~\ref{sec-SLT-grammar}, 
tree automata can be evaluated in polynomial time 
for SLT grammars, hence the same holds for the above
five formalisms.
In this section we consider another important processing primitive:
\emph{subtree equality check}.
Consider a program which realizes two independent node traversals 
of an unranked tree, using one of (1)--(5) as memory representation. 
At a given moment we wish to check if the subtrees at the two nodes
of the traversals coincide.
How expensive is this check?
As it turns out, the formalisms behave quite differently for this task.
The dags (1)--(3) as well as SL grammar-compressed dags
(4) allow efficient equality check.
We show below (Theorem~\ref{lemma:subtr_equality}) that
for an appropriate representation of the two nodes, 
this test can be performed in time $O(\log N)$, where $N$ is 
the number of tree nodes. 
For SLT grammars such a check is much more expensive.
Note that we cannot unfold the subtrees and check node by node,
as this can take exponential time.
For SLT grammars a polynomial time algorithm is known,
based on Plandowski's result~\cite{DBLP:conf/esa/Plandowski94}.
A new, fine difference between the dags (1)--(3) on the one hand
and (4) and (5) on the other hand is that we can also check equality of 
sibling sequences  in  time $O(\log N)$ for (1)--(3) (see Theorem~\ref{lemma:endseq_equality}).
For (4) and (5) we are not aware of such an algorithm.

Let $t = (V,\gamma,\lambda) \in \mathcal{T}(\Sigma)$ be an unranked tree.
Recall that the preorder traversal of $t$ ($\text{pre}(t)$ for short) enumerates the nodes of $t$
by first enumerating the root, followed by the preorder traversals of the direct subtrees of the root.
Formally, if $r$ is the root of $t$ and $\gamma(r) = v_1 \cdots v_n$, then
$\text{pre}(t) = r \; \text{pre}(t/v_1) \cdots \text{pre}(t/v_n)$.
The preorder number of a node $u \in V$ is its position in  $\text{pre}(t)$.
In what follows we identify a preorder number $p$ with
the node in $t$ that it represents, and simply speak
of ``the node $p$''.
In particular, we denote with $t/p$ ($1 \leq p \leq 
\|t\|$)
the subtree rooted at node $p$.

\begin{theorem}\label{lemma:subtr_equality}
Let $t$ be an unranked tree with $N$ nodes.
Given 
$g=\text{dag}(t)$ or
a minimal  SL grammar-compressed dag $g$ with $\text{eval}(g)=t$ (this includes the hdag) or
$g=\text{bdag}(t)$,
one can, after $O(|g|)$ time preprocessing,
check for given $1\leq p,q\leq N$
whether $t/p = t/q$ in time $O(\log N)$.
\end{theorem}

\begin{proof}
Let $t = (V,\gamma,\lambda) \in \mathcal{T}(\Sigma)$.
First, consider $g=\text{dag}(t) = (U,\gamma',\lambda')$.
For  $1 \leq p \leq N$ let $y_p$ be the unique 
node of $g$ such that $\text{eval}_g(y_p) = t/p$.
Then, $t/p = t/q$ if and only if $y_p = y_q$. Hence,
it suffices to show that
the dag-node $y_p$ can be computed from $p$ in  time
$O(\log N)$ (after $O(|g|)$ time preprocessing).
For this, we use techniques from \cite{DBLP:conf/soda/BilleLRSSW11}. 
More precisely, we construct in time $O(|g|)$ an SL string grammar 
$G'$ for the word $y_1 y_2 \cdots y_N \in U^*$. For this, we introduce for 
every node $u \in U$ of the dag $g$ a nonterminal $\hat{u}$.
If $\gamma'(u) = u_1 \cdots u_n$, then we set
$\hat{u} \to  u \hat{u}_1 \cdots \hat{u}_n$.
It is straightforward to show that this SL string grammar produces 
the word  $y_1 y_2 \cdots y_N$.  Note 
that $|G'| = |g|$.  

It now suffices to show that for a given number $1 \leq p \leq N$, the $p$-th
symbol of $\text{eval}(G')$ can be computed in time $O(\log N)$ after
$O(|g|) = O(|G'|)$ time preprocessing.
This is possible by  \cite[Theorem~1.1]{DBLP:conf/soda/BilleLRSSW11}.  
Actually, in order to apply this result, we first have to transform 
$G'$ into Chomsky normal form, which is also possible in time 
$O(|G'|) = O(|g|)$.  

For a minimal  SL grammar-compressed dag $g = (U, \gamma, \lambda, G)$ 
for $t$, where $G = (N,U,\rho)$,
essentially the same
procedure as for the dag applies. The set of nonterminals of the SL string grammar $G'$
is  $\{ \hat{u} \mid u \in U\} \cup N$. For $u \in U$ with $\gamma(u)  = \alpha_1 \cdots \alpha_n$
(with $\alpha_i \in U \cup N$) we 
set $\hat{u} \to u \hat{\alpha}_1 \cdots \hat{\alpha}_n$, where 
$\hat{\alpha}_i = \hat{v}$ if $\alpha_i = v \in U$ and $\hat{\alpha}_i = \alpha_i$ if 
$\alpha_i \in N$.  The right-hand sides for the $G'$-nonterminals from $N$ are simply copied
from the grammar $G$ with every occurrence of a symbol $u \in U$ replaced by $\hat{u}$
The reader finds an example of the construction in
Example~\ref{ex-subtree-equality-test} below.

Finally, for $g = \text{bdag}(t) = (U,\gamma,\lambda)$ we can proceed
similarly. Again we construct in time $O(|g|)$ an SL string grammar $G'$.
The set of nonterminals of $G'$ is $\{ \hat{u} \mid u \in U\}$.
For every $u \in U$ with $\lambda(u) \neq \Box$
and $\gamma(u) = u_1 u_2$ we set $\hat{u} \to u \alpha_1 \alpha_2$,
 where $\alpha_i = \varepsilon$ if $\lambda(u_i) = \Box$ and 
$\alpha_i = \hat{v}$ if $\lambda(u_i) = v \in U$. Note that 
for given preorder numbers $1 \leq p,q \leq N$, the $p$-th symbol 
of $\text{eval}(G')$ is equal to the $q$-th symbol 
of $\text{eval}(G')$ if and only if 
the sibling sequences at nodes $p$ and $q$ of $t$ are equal.
But we want to check whether the subtrees rooted at $p$ and $q$
are equal.  For this, assume that using 
\cite[Theorem~1.1]{DBLP:conf/soda/BilleLRSSW11}
we have computed in time $O(\log N)$ the $p$-th symbol $y_p \in U$ (resp.
the $q$-th symbol $y_q \in U$) of $\text{eval}(G')$. 
Then, $t/p = t/q$
is equivalent to the following conditions:
(i) $\lambda(y_p) = \lambda(y_q)$
and (ii) either $y_p$ and $y_q$ do not have left children in $g$, or 
the left children coincide.
Since these checks only require constant time, we
obtain the desired time complexity.
\qed
\end{proof}

\begin{example} \label{ex-subtree-equality-test}
Consider the following minimal SL grammar-compressed dag from Example~\ref{ex-grammar-compressed-dag}:
\begin{alignat*}{3}
A_1 & \to f(A,D,A_4,D,C) &   \qquad  A & \to a  & \qquad     D & \to A_2 A_3 \\
A_2 & \to g(E,A)              &   \qquad    B & \to b &   \qquad    E & \to AA \\
A_3 & \to h(E,B)              &   \qquad    C & \to c &   \qquad       & \\
A_4 & \to f(D)                   &   \qquad      &
\end{alignat*}
Then the construction from the proof of Theorem~\ref{lemma:subtr_equality}
yields the following SL grammar $G'$:
\begin{alignat*}{3}
\hat{A}_1 & \to A_1 \hat{A} D \hat{A}_4 D \hat{C} &   \qquad  \hat{A} & \to A  & \qquad     D & \to \hat{A}_2 \hat{A}_3 \\
\hat{A}_2 & \to A_2 E \hat{A}              &   \qquad    \hat{B} & \to B &   \qquad    E & \to \hat{A}\hat{A} \\
\hat{A}_3 & \to A_3 E \hat{B}              &   \qquad    \hat{C} & \to C &   \qquad       & \\
\hat{A}_4 & \to A_4 D                   &   \qquad      &
\end{alignat*}
This grammar produces the string
$$
\text{eval}(G') = A_1 A A_2 A^3 A_3 A^2 B A_4 A_2 A^3 A_3 A^2 B A_2 A^3 A_3 A^2 B C.
$$
\end{example}
We observe that for general SLT grammars,
a result such as the one of Theorem~\ref{lemma:subtr_equality} is not
known. To our knowledge,
the fastest known way of checking $t/p = t/q$
for a given SLT grammar $G$ for $t$ works as follows:
From $G$ we can again easily build an
SL string grammar $G'$ for the preorder traversal of $t$,
see, e.g.~\cite{DBLP:journals/is/BusattoLM08,DBLP:journals/corr/abs-1012-5696}.
Assume that the subtree of $t$ rooted in $p$ (resp., $q$) consists
of $m$ (resp., $n$) nodes. Then we have to check whether the substring
of $\text{eval}(G')$
from position $p$ to position $p+m-1$ is equal to the substring
from position $q$ to position $q+n-1$. 
Using Plandowski's result \cite{DBLP:conf/esa/Plandowski94}, 
this can be checked in time polynomial in
the size of $G'$ and hence in time polynomial in the size 
of the SLT grammar $G$.
Note that more efficient alternatives than Plandowski's
algorithm exist, see, e.g.~\cite{l12} for a survey,
but all of them require at least quadratic time in the size of the SL grammar.

In the context of XML document trees, it is also interesting
to check equivalence of two sibling sequences.
For the dag, bdag, and hdag, this problem can be solved again very
efficiently: 

\begin{theorem}\label{lemma:endseq_equality}
Let $t$ be an unranked tree with $N$ nodes.
Given $g=\text{dag}(t)$ or
$g=\text{bdag}(t)$ or
$g=\text{hdag}(t)$ 
we can, after $O(|g|)$ time preprocessing,
check for given $1\leq p,q\leq N$, whether 
$\text{sibseq}(p)=\text{sibseq}(q)$ in time
$O(\log N)$.
\end{theorem}

\begin{proof}
The result for the dag follows from the hdag-case, since
the hdag can be constructed in linear time from the dag by constructing
the minimal dag for the forest consisting of the fcns encodings of the right-hand
sides of $\mathcal{G}^{\text{red}}_{\text{dag}(t)}$
(recall that the minimal dag can be constructed in linear time \cite{DBLP:journals/jacm/DowneyST80}), and this linear time computation
is part of the preprocessing.  Furthermore,
we have already dealt with the bdag in the last paragraph of the proof of 
Theorem~\ref{lemma:subtr_equality}. Hence, it remains to consider the hdag.
We assume that the $g=\text{hdag}(t)$ is given as a minimal
SL grammar-compressed dag $D = (V, \gamma, \lambda, G)$, where
the SL grammar $G = (N, V, \rho)$ is \emph{right regular}, i.e., for every nonterminal
$X \in N$ we have $\rho(X) \in V^* N \cup V^+$, and similarly, for every
$v \in V$ we have $\gamma(v) \in V^* N \cup V^*$, see the end of Section~\ref{sec:dag_plus_string}.
After introducing additional nonterminals, we can assume that
for every $X \in N$ we have $\rho(X) \in V N \cup V$, and for every
$v \in V$ we have $\gamma(v) \in N \cup \{\varepsilon\}$ (this transformation is possible in time $O(|g|)$).
Then, the elements of $\text{sib}(t) \setminus \{t \}$ correspond to the elements
of $N$. 

We now construct an SL string grammar $G'$ as follows, see also Example~\ref{ex-hdag-sibling-test} below:
The set of nonterminals of $G'$ is 
$\{\hat{X} \mid X \in N \} \cup V$ and the set of terminals is $N$. The start nonterminal is the root $r \in V$ of 
the hdag. For every $v \in V$ we set
$$
v \to \begin{cases}
 \varepsilon & \text{ if } \gamma(v) = \varepsilon \\
 \hat{X} & \text{ if } \gamma(v) = X \in N .
\end{cases}
$$
For every $X \in N$ we set
$$
\hat{X}  \to \begin{cases}
X v \hat{Y} & \text{ if }  \rho(X) = v Y, v \in V, Y \in N \\
X v & \text{ if }  \rho(X) = v \in V.
\end{cases}
$$
Then $\text{sibseq}(p)=\text{sibseq}(q)$ holds for two numbers $1\leq
p,q\leq
\|t\|$ if and only
if $p=q=1$ or $p > 1$, $q > 1$, and the $(p-1)$-th symbol of $\text{eval}(G')$ is equal to 
the $(q-1)$-th symbol of $\text{eval}(G')$.  We deal with the case $p=q=1$ separately because
the sibling sequence $t$ corresponding to the root of $t$ is not represented in $\text{eval}(G')$
(the latter string has length $N-1$).
\qed
\end{proof}

\begin{example} \label{ex-hdag-sibling-test}
Consider the following hdag (our running example from Section~\ref{sec-hdag}), written
as an SL grammar-compressed dag of the form used in the proof of Theorem~\ref{lemma:endseq_equality}.
\begin{alignat*}{2}
S&\to f(X_0),   \qquad &  X_0 & \to B X_1, \\
B &\to f(X_1),              & X_1 & \to A X_2, \\
A &\to g(X_3),              & X_2 & \to A, \\
C & \to  a                    & X_3 & \to C .
\end{alignat*}
It produces the tree $t = f(f(g(a),g(a)),g(a),g(a))$, see Figure~\ref{fig:hdag}. The nonterminal $X_0$
represents the sibling sequence $f(g(a),g(a)) g(a) g(a)$, $X_1$ 
represents $g(a) g(a)$, $X_2$ represents $g(a)$, and $X_3$ represents $a$.
These are all sibling sequences except for the length-$1$ sequence $t$.

According to the construction from the proof of Theorem~\ref{lemma:endseq_equality},
we obtain the following SLT grammar $G'$:
\begin{alignat*}{2}
S & \to \hat{X}_0,  \qquad   & \hat{X}_0  & \to  X_0 B \hat{X}_1, \\
B &\to \hat{X}_1,              & \hat{X}_1 & \to X_1 A \hat{X}_2, \\
A &\to \hat{X}_3,              & \hat{X}_2 & \to X_2 A, \\
C & \to  \varepsilon                    & \hat{X}_3 & \to X_3 C .
\end{alignat*}
It produces the string
$$
\text{eval}(G') = X_0 X_1 X_3 X_2 X_3 X_1 X_3 X_2 X_3.
$$
For instance, at preorder positions $3$ and $7$ the same sibling sequence, namely $g(a)g(a)$ 
starts in the tree $t$. This sibling sequence is represented by the symbol $X_1$. Indeed, the 2nd
and $6$-th symbol in  $\text{eval}(G')$ is $X_1$.
\end{example}
%
For an SL grammar-compressed dag, the best solution for checking
$\text{sibseq}(p)=\text{sibseq}(q)$ we are
aware of uses again an equality check for SL grammar-compressed strings.

\begin{table*}[ht]
\centering
\begin{tabular}{lrrrrrr}    
\toprule
 File & Edges & mD & aC & mC & dag & bdag \\
\midrule

1998statistics  & 28305    & 5  & 22.4 & 50     & 1377    & 2403  \\ 
catalog-01      & 225193   & 7  & 3.1  & 2500   & 8554    & 6990 \\ 
catalog-02      & 2240230  & 7  & 3.1  & 25000  & 32394   & 52392 \\
dblp            & 3332129  & 5  & 10.1 & 328858 & 454087  & 677389\\ 
dictionary-01   & 277071   & 7  & 4.4  & 733    & 58391   & 77554 \\ 
dictionary-02   & 2731763  & 7  & 4.4  & 7333   & 545286  & 681130 \\
EnWikiNew       & 404651   & 4  & 3.9  & 34974  & 35075   & 70038 \\ 
EnWikiQuote     & 262954   & 4  & 3.7  & 23815  & 23904   & 47710 \\ 
EnWikiVersity   & 495838   & 4  & 3.8  & 43593  & 43693   & 87276 \\ 
EnWikTionary    & 8385133  & 4  & 3.8  & 726091 & 726221  & 1452298 \\ 
EXI-Array       & 226521   & 8  & 2.3  & 32448  & 95584   & 128009  \\ 
EXI-factbook    & 55452    & 4  & 6.8  & 265    & 4477    & 5081 \\ 
EXI-Invoice     & 15074    & 6  & 3.7  & 1002   & 1073    & 2071 \\ 
EXI-Telecomp    & 177633   & 6  & 3.6  & 9865   & 9933    & 19808 \\ 
EXI-weblog      & 93434    & 2  & 11.0 & 8494   & 8504    & 16997 \\ 
JSTgene.chr1    & 216400   & 6  & 4.8  & 6852   & 9176    & 14606 \\
JSTsnp.chr1     & 655945   & 7  & 4.6  & 18189  & 23520   & 40663 \\ 
medline         & 2866079  & 6  & 2.9  & 30000  & 653604  & 740630 \\ 
NCBIgene.chr1   & 360349   & 6  & 4.8  & 3444   & 16038   & 14356 \\ 
NCBIsnp.chr1    & 3642224  & 3  & 9.0  & 404692 & 404704  & 809394 \\ 
sprot39.dat     & 10903567 & 5  & 4.8  & 86593  & 1751929 & 1437445 \\ 
SwissProt       & 2977030  & 4  & 6.7  & 50000  & 1592101 & 1453608 \\ 
treebank        & 2447726  & 36 & 2.3  & 2596   & 1315644 & 1454520 \\ 
 \bottomrule
\end{tabular}
\caption{The XML documents in Corpus I, their characteristics,
and dag/bdag sizes}
\label{list_1}
\end{table*}

\begin{table*}[ht]
\centering
\begin{tabular}{lrrrrr}    
\toprule
 File & rbdag & hdag & rhdag & DS & TR \\
\midrule

1998statistics   & 2360    & 1292    & 1243    & 561     & 501 \\ 
catalog-01       & 10303   & 4555    & 6421    & 4372    & 3965\\ 
catalog-02       & 56341   & 27457   & 29603   & 27242   & 26746\\
dblp             & 681744  & 358603  & 362571  & 149964  & 156412 \\ 
dictionary-01    & 75247   & 47418   & 46930   & 32139   & 22375 \\ 
dictionary-02    & 653982  & 414356  & 409335  & 267944  & 167927 \\
EnWikiNew        & 70016   & 35074   & 35055   & 9249    & 9632 \\ 
EnWikiQuote      & 47690   & 23903   & 23888   & 6328    & 6608 \\ 
EnWikiVersity    & 87255   & 43691   & 43676   & 7055    & 7455 \\ 
EnWikTionary     & 1452270 & 726219  & 726195  & 81781   & 84107 \\ 
EXI-Array        & 128011  & 95563   & 95563   & 905     & 1000 \\ 
EXI-factbook     & 2928    & 3847    & 2355    & 1808    & 1392 \\ 
EXI-Invoice      & 2068    & 1072    & 1069    & 96      & 108 \\ 
EXI-Telecomp     & 19807   & 9933    & 9932    & 110     & 140 \\ 
EXI-weblog       & 16997   & 8504    & 8504    & 44      & 58 \\ 
JSTgene.chr1     & 14114   & 7901    & 7271    & 3943    & 4208 \\
JSTsnp.chr1      & 37810   & 22684   & 19532   & 9809    & 10327 \\ 
medline          & 381295  & 466108  & 257138  & 177638  & 123817 \\ 
NCBIgene.chr1    & 10816   & 11466   & 7148    & 6283    & 5166 \\ 
NCBIsnp.chr1     & 809394  & 404704  & 404704  & 61      & 83 \\ 
sprot39.dat      & 1579305 & 1000376 & 908761  & 335756  & 262964 \\ 
SwissProt        & 800706  & 1304321 & 682276  & 278915  & 247511\\ 
treebank         & 1244853 & 1250741 & 1131208 & 1121566 & 528372 \\ 
 \bottomrule
\end{tabular}
\caption{The compressed sizes of the documents.}
\label{list_2}
\end{table*}

\section{Experiments}\label{sec:exp_results}

In this section we empirically compare the sizes of 
different dags of unranked trees, namely 
dag, bdag, rbdag, hdag, and rhdag.
We also include a comparison with SL grammar-compressed dags 
with RePair \cite{DBLP:conf/dcc/LarssonM99} as the string compressor,
as explained in Section~\ref{sec:dag_plus_string}, and
with TreeRePair \cite{lohmanmen13}. 
We are interested in the tree structure only, hence we did not compare
with XML file compressors like 
Xmill~\cite{DBLP:conf/sigmod/LiefkeS00}
or XQueC~\cite{DBLP:journals/toit/ArionBMP07}.

\subsection{Corpora}

We use three corpora of XML files for our tests. 
For each XML document we consider the unranked tree of
its element nodes; we ignore all other nodes such as texts,
attributes, etc.
One corpus (\emph{Corpus I}) consists of XML documents  
that have been collected from the web, and which have often
been used in the context of XML compression research,
e.g., in~\cite{DBLP:conf/vldb/KochBG03,DBLP:journals/is/BusattoLM08,lohmanmen13}.
Each of these files is listed in Table~\ref{list_1}
together with the following characteristics:
number of edges, maximum depth (mD), average number of children of a node (aC),
and maximum number of children of a node (mC).
Precise references to the origin of these files can be 
found in~\cite{lohmanmen13}. 
The second corpus (\emph{Corpus II}) consists of all 
well-formed XML document trees with more than 10,000 edges and a 
depth of at least four from the 
\textit{University of Amsterdam XML Web 
Collection}\footnote{ http://data.politicalmashup.nl/xmlweb/}. 
We decided on fixing a minimum size because there is no necessity 
to compress documents of very small size, and we chose a minimum
depth because our subject is tree compression rather than list compression.
Note that out of the over 180,000 documents of the collection,
only 1,100 fit our criteria and are part of Corpus~II (more than $27,000$ were 
ill-formed and more than $140,000$ had less than $10,000$ edges).
The documents in this corpus are somewhat smaller than
those in Corpus~1, but otherwise have similar
characteristics (such as maximal depth and average number of children)
as can be seen in Table~\ref{characteristics_average}.
The third corpus (\emph{Corpus III}) consists of term rewriting 
systems\footnote{http://www.termination-portal.org/wiki/TPDB}.
These are stored in XML files, but, are rather atypical 
XML documents, because their tree structures are trees
with small rank, i.e., there are no long sibling sequences.
This can be seen in Table~\ref{characteristics_average}, which shows that
the average number of children is only $1.5$ for these files.

\begin{table}
\centering
\begin{tabular}{lrrrrr} 
\toprule
Corpus& Edges & mD   & aC & mC   \\
\midrule
I & 1.9 $\cdot 10^6$ & 6.6 & 5.7 &  8 $\cdot 10^4 $\\
II & 79465 & $7.9$ & $6.0$ & 2925  \\
III & 1531 & 18 & 1.5 & 13.2  \\
\bottomrule
\end{tabular}
\caption{\label{characteristics_average} Document characteristics, average values.}
\end{table}

\subsection{Experimental setup}

For the dag, bdag, rbdag, and hdag we built our own implementation.
It is written in C++  (g++ version 4.6.3 with O3-switch) 
and uses Libxml 2.6 for XML parsing.
It should be mentioned that these are only rough prototypes and
that our code is not optimized at all. The running times listed in
Table~\ref{times} should be understood with this in mind. 
For the RePair-compressed dag we use 
Gonzalo Navarro's implementation of 
RePair\footnote{ http://http://www.dcc.uchile.cl/$\sim$gnavarro/software/}. This is called ``DS'' in our tables.
For TreeRePair, called ``TR'' in the tables, we use
Roy Mennicke's implementation\footnote{ http://code.google.com/p/treerepair/}
and run with max\_rank=1, which produces 1-SLT grammars.
Our test machine features an Intel Core i5 with 2.5Ghz and 4GB
of RAM. 

\begin{table}[t]
\centering
\begin{tabular}{lrrrrrrrr}
\toprule
Corpus & Parse & dag & hdag & DS & TR \\
\midrule
I & 35 & 43 & 46 & 48 & 175\\
II & 85 & 105 & 120 & 117 & 310\\
III & 6.9 & 8.7 & 9.2 & 10.0 & 14.8\\
\bottomrule
\end{tabular}
\caption{Cumulative Running times (in seconds).}
\label{times}
\end{table}

\subsection{Comparison}\label{sec:exp_results_dag}

Consider first Corpus~1 and the numbers shown in
Table~\ref{list_1} and~\ref{list_2}. The most interesting file, 
concerning the effectiveness of
the hybrid dag and of the reverse binary encoding, 
is certainly the medline file.
Just like dblp, it contains bibliographic data.
In particular, it consists of MedlineCitation elements;
such elements have ten children, the last of which varies
greatly (it is a MeshHeadingList node with varying children
lists) and thus cannot be shared in the dag. 
This is perfect for the reverse hybrid dag, which first
eliminates repeated subtrees, thus shrinking the number
of edges to 653,604, and then applies the last child/previous
sibling encoding before building a dag again.
This last step shrinks the number of edges to impressive
257,138. In contrast, the reverse binary dag has a size 
of 381,295. Thus, for this document really the combination of 
both ways of sharing, subtrees and reversed sibling sequences,
is essential. We note that in the context of the first attempts
to apply dag compression to XML~\cite{DBLP:conf/vldb/KochBG03} 
the medline files were
particularly pathological cases where dag compression did not work well.
We now have new explanations for this: using \emph{reverse} 
(last child/previous sibling) encoding slashes the size of the dag by almost
one half. And using hybrid dags again brings an improvement of
more than $30\%$.
The dblp document is similar, but does not make use of optional
elements at the end of long sibling lists. Thus, the reverse dags
are not smaller for dblp, but the hybrid dag is indeed more than
$20\%$ smaller than the dag.
The treebank document, which is a ranked tree and does not
contain long lists, gives hardly any improvement of hybrid
dag over dag, but the reverse hybrid dag is somewhat smaller
than the dag (by $5\%$). 
For treebank, TreeRePair is unchallenged and
produces a grammar that is less than half the size of DS.

\begin{table*}[t]
\centering
\small
\begin{tabular}{lrrrrrrrrrr}
\toprule
C. & Input & dag & bdag & rbdag & hdag & G(hd) & rhdag & G(rh) & DS & TR \\
\midrule
I & 43021 & 7815 & 9292 & 8185 & 6270 & 6323 & 5220 & 5285 & 2523 & 1671 \\ 
II & 90036 & 13510 & 15950 & 14671 & 10884 & 11109 & 9806 & 10039 & 5162 & 3957 \\ 
III &  2095 & 354 & 391  & 390 & 319 & 362  & 320 & 364 & 324 & 310 \\ 
\bottomrule
\end{tabular}
\caption{\label{table:overview} Accumulated sizes (in thousand edges). {\em C} stands for
{Corpus}, {\em G(hd)} for the grammar size of the hdag and {\em G(rh)} for the grammar size of 
the reverse hybrid dag.}
\end{table*}

Next, consider the accumulated numbers for the three
corpora in Table~\ref{table:overview}.
For Corpus~I, the reverse hdag is smaller than the dag
by around $38\%$ while the hdag is only around $25\%$ smaller
than the dag. As noted in Section~\ref{section_reverseDag}, the somewhat 
surprising outcome that the reverse binary encoding enables better
compression results from the custom that in many XML-documents optional
elements are listed last. This means that there are more common prefixes than
suffixes in child sequences. Hence the reverse schemes perform better.
When we transform hdags into
SLT grammars (according to Section~\ref{sec:dag_plus_string}), 
then we get a modest size increase of about 1--$2\%$.
For the term rewriting systems of Corpus~III,
the hdag improves about $10\%$ over the dag. Represented as
grammars, however, this improvement disappears and in fact
we obtain an accumulated size that is slightly larger than
the dag. Note that for this corpus, also TreeRePair (TR) is not
much smaller than the dag, and DS is even smaller than TR.
 Compared to the dag, TreeRePair
shares tree patterns (=connected subgraphs). Hence, the trees
in Corpus~III do not contain many repeated tree patterns which
are not already shared by the dag.
When we compare DS with TR, then we see on corpora~I and~II
that TreeRePair grammars are on average around $34\%$ smaller
than DS, while on Corpus~III it is only $23\%$ smaller.
On very flat files, such as the EXI documents in 
Table~\ref{list_1}, DS is about as good as TreeRePair.
For combined dag and string compression we also
experimented with another grammar-based string
compressor: Sequitur~\cite{DBLP:journals/jair/Nevill-ManningW97},
but found the combined sizes to be larger than with RePair.
Concerning running times (see Table~\ref{times}) note that 
the dag-variants stay close to the parsing time, while TreeRePair needs 
considerably more time. Hence, dags should be used when 
light-weight compression is preferred.

\section{Conclusion and future work}

We compare the sizes of five different formalisms for 
compactly representing unranked trees:
\begin{enumerate}[(1)]
\item dag
\item binary dag
\item hybrid dag
\item combination of dag and SL grammar-based string compression
\item SLT grammars (e.g. produced by BPLEX or TreeRePair)
\end{enumerate}
For the comparison of (1)--(3) we prove precise bounds:
(i) the size of the binary dag of a tree is bounded by twice the size of the hybrid dag 
of the tree and 
(ii) the size of the unranked dag of a tree is bounded by the square of the size 
of the hybrid dag of the tree. As a corollary we obtain that 
the size of the dag is at least of the size of the
binary dag, and at most the square of the size of the binary dag.
We also prove that for (1)--(4), checking equality of the subtrees rooted
at two given nodes of these structures, can be carried out in 
$O(\log N)$ time, where $N$ is the number of nodes of the tree.
One advantage of binary and hybrid dags, is that they also support
the efficient checking of equality of (ending) sibling sequences
in $O(\log N)$ time.

Our experiments over two large XML corpora and one corpus consisting of term
rewriting systems show that dags and binary dags are the quickest to construct.
Out of the dags (1)--(3), the reverse hdag
(which uses a last child/previous sibling encoding) gives the smallest results.
On our XML corpora, 
using the reverse binary encoding instead of the
standard first child/next sibling encoding gives a compression improvement
of more than 20\%.
As a practical yardstick we observe: For applications where sibling sequence check
is important, or where the uniqueness of the compressed structures is important,
the hybrid dag is a good choice. If
strong compression is paramount, then structures (4) and (5) are appropriate.
The advantage of (4) over (5) is its support of efficient subtree equality test.

We give generating functions for the exact average sizes of dags and bdags over
unranked trees on $n$ nodes and $m$ labels. 
We show that asymptotically the expected edge sizes of dags and bdags
coincide, and that the node size of bdags is twice as large as the 
node size of dags.  

In future work we would like extend our average-case size analysis also to
the hybrid dag and to Repair-compressed trees. 
Further, we would like to apply our compression within other recent 
compression schemes in databases, such as 
for instance factorized databases~\cite{DBLP:journals/pvldb/BakibayevOZ12}.

\begin{acknowledgements}
The second and fourth author were supported by the DFG grant
LO 748/8.
The third author was supported by the DFG grant
INST 268/239 and by the 
Engineering and Physical Sciences Research Council 
project "Enforcement of Constraints on XML streams" 
(EPSRC EP/G004021/1).
\end{acknowledgements}

\appendix

\section{Appendix:  Proofs 
%
of the asymptotic results}

In this appendix we prove the results from Section~\ref{subsection:asymptotic}.
They rely on methods described in the book {\em Analytic Combinatorics}~\cite{AnalyticCombinatorics}
from 
Flajolet and Sedgewick;
 in particular, the {\em Transfer Theorem}
in Chapter 6.

\subsection{Proof of Proposition~\ref{lemme1-fl}
(case $m=1$)}
\label{proof:fl}
For notational simplicity, we 
first show the proof for unlabeled trees (i.e.\ a singleton alphabet, $m=1$).
In Section~\ref{sec:binary-m} we 
adapt the proof 
to arbitrary alphabets.
By Lemma~\ref{lem:C-binary}
and Proposition~\ref{lemma:binary_K},  the \gf\ for the cumulated node
size of compressed  binary trees reads
$$
\mathbf \KN (z): =\mathbf  \KN _1^{\mathcal{B}}(z) =\sum_{p\ge 0} B_p \mathbf C_p(z)
$$
with
\beq\label{Cp1}
B_p=\frac 1{p+2}{2p+2 \choose p+1} \quad \hbox{and} \quad 
\mathbf C_p(z)= \frac{\sqrt{1-4z+4z^{p+2}} -\sqrt{1-4z}}{2z^2}.
 \eeq
We want to prove Proposition~\ref{lemme1-fl}, which, when $m=1$, reads as follows.
\begin{proposition}\label{prop:m=1}
The generating function $\mathbf \KN (z)$
is analytic in the domain $D$ 
defined by $|z|< \frac{1}{2}$ and $z \notin
[ \frac{1}{4},\frac{1}{2} ]$. As $z$ tends to $ \frac{1}{4}$ in $D$,
one has
\begin{equation*}
\mathbf  \KN (z) = \frac{8 \kappa} {\sqrt{(1-4 z) \ln ((1-4z)^{-1})}} + O \left(\frac{1}{\sqrt{(1-4z)\ln^3((1-4z)^{-1})}} \right).
\end{equation*}
where $\kappa = \sqrt{\frac{\ln 4}{\pi}}$ .
\end{proposition}

   The proof of 
this proposition 
works in several steps.  We first show that the function
   $\mathbf \KN (z)$ is analytic in  $D$. 
   Then we split  $\mathbf \KN (z)$  into three parts.
   The splitting  depends on a threshold integer $n\equiv n(z)$.
   We treat each part independently: we estimate precisely one of
   them, and  simply bound the
   other two.  We then show that these two parts can be neglected for a
   suitable choice of $n(z)$.
   \paragraph{Step 1: $\mathbf \KN (z)$ is analytic in $D$.} We first
   prove that  the series $\mathbf C_p(z)$ defined by~\eqref{Cp1} is analytic in   $D$. This is implied by the following lemma. 

\begin{lemma} \label{lem:rouche}
Let $\Delta_p(z)=1-4z+4z^{p+2}$. For $p \ge 1$, this
      polynomial has exactly one root of modulus less than $1/2$. This
      root is real and larger than $1/4$. 
      When $p=0$, then $\Delta_p(z)=(1-2z)^2$,
with a double root at $1/2$. 
In particular, $\Delta_p(z)$ does not vanish in $D$.

More generally, the rational function $u_p:= 4z^{p+2}/(1-4z)$ does not
take any value from $(-\infty, -1]$ for $z\in D$. 
\end{lemma}
\begin{proof} To prove the existence of a unique root of modulus less
      than $1/2$, we apply Rouch\'e's theorem (see e.g. \cite[p.~270]{AnalyticCombinatorics})
by comparing  $\Delta_p(z)$ to the auxiliary
      function $1-4z$, and using the integration contour 
      $\{|z|=1/2\}$. Details are very simple, and are left to the 
      reader. 

Since $\Delta_p$ has real coefficients, the unique root of modulus
less than $1/2$ must be real. One then concludes the proof by observing that
$\Delta_p(1/4)>0$ while $\Delta_p(1/2)<0$.

The proof of the final statement is similar, upon comparing the
polynomials $c(1-4z)+4z^{p+2}$ and $c(1-4z)$, for $c\ge 1$.
\qed
\end{proof}
In order to prove that $\mathbf \KN (z)$ itself is analytic in the domain $D$, we
      rewrite it as follows, still denoting $u_p=4z^{p+2}/(1-4z)$:
      \begin{eqnarray}
\mathbf \KN (z) 
&=&\frac{\sqrt{1-4z}}{2z^2} \sum_{p\ge 0}
B_{p} \left( \sqrt{1+u_p}-1\right)\nonumber
\\
&=&\frac{\sqrt{1-4z}}{2z^2} \sum_{p\ge 0}
B_{p} \left(\frac{u_p}2+ \sqrt{1+u_p}-1-\frac{u_p}2\right)\nonumber
\\
&=&\frac{1}{\sqrt{1-4z}}\sum_{p\ge 0} B_p z^p + \frac{\sqrt{1-4z}}{2z^2} \sum_{p\ge
0}B_{p} \left( \sqrt{1+u_p}-1-\frac{u_p}2\right)\nonumber
\\
&=&\frac{1}{\sqrt{1-4z}}
\mathbf B(z)
+ \frac{\sqrt{1-4z}}{2z^2} \sum_{p\ge 0}B_p \left(\sqrt{1+u_p}-1-\frac{u_p}2\right),\label{K-final}
      \end{eqnarray}
where
\begin{equation}\label{B-def}
\mathbf B(z)= \sum_{p\ge 0}B_p z^p=\frac{1-2z-\sqrt{1-4z}}{2z^2}
\end{equation}
is the generating function of the shifted Catalan numbers.

Since $\sqrt{1-4z}$ and $\mathbf B(z)$ are analytic in $D$, it suffices to
      study the convergence of the sum over $p$ in~\eqref{K-final}.
%
    \begin{lemma}\label{lem:b3}
      For all $u\in \mathbb{C}\setminus(-\infty, -1]$, we have
$$
\left|\sqrt{1+u}-1-\frac u 2\right| \le \frac{|u|^2}2.
$$
      \end{lemma}
      \begin{proof}
 Write $x=\sqrt{1+u}$, so that $\Re(x)>0$. Then
      $|u|=|x^2-1|=|x-1||x+1|$,  so that the above inequality reads
$$
\frac{|x-1|^2}2 \le     \frac{|x-1|^2|x+1|^2}2 ,
$$
or equivalently
$1\le|x+1|$,
which is clearly true since $\Re(x)>0$.
\qed
      \end{proof}
 Recall that all points of $D$ have modulus less than $1/2$.
Consider a disk included in $D$.  For $z$ in this disk, the quantity 
$$
|u_p|^2= \frac{16 |z|^{2p+4}}{|1-4z|^2}
$$
is uniformly bounded by $c r^{p}$, for constants $c$ and
      $r<1/4$. 
By Lemma~\ref{lem:rouche}, we can apply   Lemma~\ref{lem:b3} to $u_p$.  Hence 
  \begin{eqnarray*}
      \sum_{p\ge 0}B_p
\left|\sqrt{1+u_p}-1-\frac{u_p}2\right|
&\le&  \frac 1 2 \sum_{p\ge 0}B_p |u_p|^2\\
&\le&  \frac c 2 \sum_{p\ge 0}B_p r^p <\infty.
  \end{eqnarray*}
Hence the series occurring in~\eqref{K-final}
      converges uniformly in the disk, and
      $\mathbf \KN (z)$ is analytic in~$D$. 
\qed

\paragraph{Step 2: splitting $\mathbf \KN (z)$.}
We fix an integer $n$ and write
\begin{equation}\label{K-split}
\mathbf \KN (z)=\mathbf \KN ^{(1)}(n,z)+\mathbf \KN ^{(2)}(n,z)+\mathbf \KN ^{(3)}(n,z),
\end{equation}
where
\begin{eqnarray}
\mathbf \KN ^{(1)}(n,z)&=& \frac{\sqrt{1-4z}}{2z^2} \sum_{p=0}^n
B_{p} \left( \sqrt{1+u_p}-1\right),\nonumber
\\
\mathbf \KN ^{(2)}(n,z)&=&\frac 1{\sqrt{1-4z}} \left( \mathbf B(z)- \sum_{p=0}^n B_p z^p\right),\label{K2}
\\
\mathbf \KN ^{(3)}(n,z)&=& \frac{\sqrt{1-4z}}{2z^2} \sum_{p>n}B_{p} \left(\sqrt{1+u_p}-1-\frac{u_p}2\right).\label{K3}
\end{eqnarray}
One readily checks that~\eqref{K-split} indeed holds. Moreover, each
$\mathbf \KN ^{(i)}(n,z)$ is analytic in $D$  for any $i$ and $n$.

\paragraph{Step 3: an upper bound on $\mathbf \KN ^{(1)}$.} 
\begin{lemma}\label{lem:b1}
For any $u\in \mathbb{C} \setminus(-\infty, -1 ]$, we have $|\sqrt{1+u}-1|\le \sqrt{|u|}$.
\end{lemma}
\begin{proof}
    The proof is similar to the proof of Lemma~\ref{lem:b3}. Write $x=\sqrt{1+u}$, so that $\Re(x)>0$. Then
the inequality we want to
prove reads $|x-1|\le |x+1|$, which is clearly
true.\qed
\end{proof}

\begin{lemma}\label{lem:sum-div}
Let $y>0$, and define 
$$
a(n,y):=\sum_{p=0}^n B_py^p.
$$
For any $c>1/4$, there exists a neighborhood of $c$ such that
$$
a(n,y) = O((4y)^nn^{-3/2}),
$$
uniformly in $n$ and $y$ in this neighborhood of $c$.
\end{lemma}
\begin{proof}  Let us form the generating function of the numbers $a(n,y)$. One finds:
$$
\sum_{n\ge 0} a(n,y) x^n = \sum_{n\ge 0}\sum_{p=0}^n B_py^p x^n
= \frac{\mathbf B(xy)}{1-x} = \frac{1-2xy-\sqrt{1-4xy}}{2x^2y^2(1-x)}.
$$
Thus
$$
a(n,y)= 
\frac {1-2y}{2y^2}- [x^n] \frac{\sqrt{1-4xy}}{2x^2y^2 (1-x)}
=\frac 1{2y^2}\left( 1-2y- I(n,y)\right),
$$
where 
$$
I(n,y)=
[x^{n+2}] 
\frac{\sqrt{1-4xy}}{1-x}.
$$
We now estimate $I(n,y)$ using a Cauchy integral:
$$
I(n,y)=\frac 1 {2i\pi} \int_\circlearrowleft \frac{\sqrt{1-4xy}}{1-x} \frac{dx}{x^{n+3}},
$$
where the integral is over a small circle around the origin. Recall
that $y$ is taken in a small 
neighborhood of $c>1/4$. Hence the
dominant singularity of the integrand is at $x=1/(4y)$. (A second
singularity occurs later at $x=1$.)  Writing $x=u/(4y)$ gives
$$
I(n,y)=\frac {(4y)^{n+3}} {2i\pi} \int_\circlearrowleft  \frac{\sqrt{1-u}}{4y-u} \frac{du}{u^{n+3}},
$$
the integral being again over a small circle around the origin. We now
deform the integration contour as discussed
in~\cite[p.~382]{AnalyticCombinatorics}. This gives
$$
I(n,y)= \frac{(4y)^{n+3} }{2i\pi n^{3/2}} \int_{\mathcal{H}} \frac{\sqrt{-t}}{4y-1-t/n}
\left(1+\frac t n\right)^{-n-3} dt,
$$
where $\mathcal{H}$ is the Hankel contour at distance $1$ from the real
positive axis described in~\cite[p.~382]{AnalyticCombinatorics}.
Hence the lemma will be proved if we can show that the integral over
$\mathcal{H}$ is $O(1)$, uniformly in $n$ and $y$. As in our favorite reference~\cite[p.~383]{AnalyticCombinatorics}, we
split it into two parts, depending on whether $\Re(t) \le \ln^2 n$ or
$\Re(t) \ge \ln^2 n$. On the first part, $4y-1-t/n$ remains uniformly
away from $0$, while
$$
 \int_{\mathcal{H}\cap\{\Re(t) \le \ln^2 n \}}\left| {\sqrt{-t}}\left(1+\frac t n\right)^{-n-3}\right| dt
$$
can be shown to be $O(1)$, using the techniques of~\cite[p.~383]{AnalyticCombinatorics} or \cite{Fl-Od}. On the second part,
$|4y-1-t/n|$ reaches its minimal value $1/n$ when
$\Re(t)=n(4y-1)$. Writing $t=x+i$, we can bound the modulus of the second part by
\begin{eqnarray*}
 2 n \int_ {x\ge \ln^2 n }\sqrt{2x}\left(1+\frac x n\right)^{-n-3}
dx
&\le &
2n \exp(-\ln^2(n))\int_ {x\ge \ln^2 n }\sqrt{2x}\left(1+\frac x
n\right)^{-3} dx\\
&\le& 2n^{3/2} \exp(-\ln^2(n))\int_{u\ge 0} \sqrt{2u}(1+u)^{-3}
du\\
&=&o(n^{-\alpha})
\end{eqnarray*}
for any real $\alpha$.
This completes the proof of the lemma.\qed
\end{proof}

Combining 
Lemmas~\ref{lem:b1} and~\ref{lem:sum-div}
(with $c=1/2$)
gives,  for $z$ in $D$ close enough to  $1/4$:
\begin{eqnarray}
|\mathbf \KN ^{(1)}(n,z)|&\le& \frac{\sqrt{|1-4z|}}{2|z^2|} \sum_{p=0}^n
B_{p} \sqrt{|u_p|}\nonumber
\\
&\le & \sum_{p=0}^n
B_{p} |z|^{(p-2)/2}\nonumber\\
&=&   O\left( 4^n|z|^{n/2}n^{-3/2}\right),\label{K1-bound}
\end{eqnarray}
uniformly in $n$ and $z$.

     
 \paragraph{Step 4: an upper bound on $\mathbf \KN ^{(3)}$.}

Let us combine Lemma~\ref{lem:b3} with the expression~\eqref{K3} of
$\mathbf \KN ^{(3)}(n,z)$. This gives:
\begin{eqnarray}
\mathbf \KN ^{(3)}(n,z) &\le & \frac{4z^2}{|1-4z|^{3/2}} \sum_{p>n} B_{p} |z|^{2p}
\nonumber
\\
&\le  & \frac{4z^2}{|1-4z|^{3/2}} \sum_{p>n} B_n 4^{p-n} |z|^{2p}
\hskip 25mm (\hbox{since } B_{p+1} \le 4B_p)
\nonumber\\
&\le  & \frac{16 |z|^{2n+4}}{|1-4z|^{3/2}} B_n \sum_{p>n}  4^{p-n-1}
|z|^{2(p-n-1)}
\nonumber\\
&=& O\left(\frac{|z|^{2n}}{|1-4z|^{3/2}}B_n\right) \label{K3-bound}
\end{eqnarray}
uniformly in $n$ and $z$, for $z$ in $D$ close enough to  $1/4$.

\paragraph{Step 5: estimate of $\mathbf \KN ^{(2)}$.}
For $z$ fixed, let us determine the generating function of the numbers
$\mathbf \KN ^{(2)}(n,z)$, given by~\eqref{K2}. We find:
\begin{eqnarray*}
\sum_{n\ge 0} \mathbf \KN ^{(2)}(n,z) x^n &=& \frac 1 {\sqrt{1-4z}
} \frac{\mathbf B(z)-\mathbf B(xz)}{1-x}
\end{eqnarray*}
where $\mathbf B(z)$ is defined by~\eqref{B-def}. Equivalently,
\begin{eqnarray*}
\sum_{n\ge 0}\mathbf  \KN ^{(2)}(n,z) x^n
%
&=& \frac 1 {2x^2z^2\sqrt{1-4z}}\left((1+x) \sqrt{1-4z}-1-x+2xz
+ \frac{4z}{\sqrt{1-4xz}+\sqrt{1-4z}}\right).
\end{eqnarray*}
      Thus 
$$
\mathbf \KN ^{(2)}(n,z)= \frac 2{z\sqrt{1-4z}} I(n,z)
$$
where
$$
I(n,z)=[x^{n+2}] \frac{1}{\sqrt{1-4xz}+\sqrt{1-4z}}.
$$
Using the same principles as in the proof of Lemma~\ref{lem:sum-div}, we
now estimate $I(n,z)$ using a Cauchy integral: 
$$
I(n,z)= \frac
1{2i\pi} \int_\circlearrowleft \frac{1}{\sqrt{1-4xz}+\sqrt{1-4z}} \frac{dx}{x^{n+3}}
$$
where the integral is over a small circle around the origin. The
unique singularity of the integrand is at $x=1/(4z)$. Writing
$x=u/(4z)$ gives
$$
I(n,z)= \frac {(4z)^{n+2}}{2i\pi} \int _ \circlearrowleft\frac 1
{\sqrt{1-u}+\sqrt{1-4z}} \frac{du}{u^{n+3}}. 
$$
Changing the integration contour into a Hankel one, gives, as in
Lemma~\ref{lem:sum-div},
$$
I(n,z)= \frac {(4z)^{n+2}}{2i\pi \sqrt n} \int_{\mathcal{H}} \frac 1
{\sqrt{-t}+\sqrt{n(1-4z)}}\left( 1+ \frac t n\right)^{-n-3} {dt}.
$$
Again, we split  the integral into two parts, depending on whether
$\Re(t) \le \ln^2 n$ or 
$\Re(t) \ge \ln^2 n$. We assume moreover that
\begin{equation}\label{cond}
n\rightarrow \infty \quad \hbox{and} \quad n(1-4z)\rightarrow 0.
\end{equation}
 The first part of the integral is then 
$$
 \frac {(4z)^{n+2}}{\Gamma(1/2) \sqrt n } \left(1 + O(\sqrt{n(1-4z)})
 +O(n^{-1})\right),
$$
while the second part is found to be smaller than $(4z)^n
 n^{-\alpha}$, for any real $\alpha$.
 Hence
$$
I(n,z)=  \frac{(4z)^{n+2}}{\sqrt\pi \sqrt n} \left(1 + O(\sqrt{n(1-4z)})
 +O(n^{-1})\right),
 $$
so that
\begin{equation}\label{K2-est}
\mathbf \KN ^{(2)}(n,z) = 
\frac {8 (4z)^{n+1} } 
{\sqrt\pi\sqrt{1-4z}\sqrt n} \left(1 + O(\sqrt{n(1-4z)})
 +O(n^{-1})\right).
\end{equation}

\paragraph{Step 6: the threshold $n(z)$.} We finally want to
correlate $n\equiv n(z)$ and $z$ so that, as $z$ tends to $ 1/4$ (in
the domain $D$),  the function $\mathbf \KN (z)$ is dominated by $\mathbf \KN ^{(2)}(n,z)$. Given
the bounds~\eqref{K1-bound} and~\eqref{K3-bound} on $\mathbf \KN ^{(1)}(n,z)$ and
$\mathbf \KN ^{(3)}(n,z)$, the estimate $B_n \sim 4^n n^{-3/2}$ (up to a
multiplicative constant), and the estimate~\eqref{K2-est} of $\mathbf \KN ^{(2)}$, we want
$$
\frac{4^n |z|^{n/2}}{ n^{3/2}} = o\left(  \frac{(4z)^{n}}{\sqrt{1-4z} \sqrt n}\right)
$$
and 
$$
\frac{4^n |z|^{2n}}{ n^{3/2}|1-4z|^{3/2}} =
o\left(  \frac{(4z)^{n}}{\sqrt{1-4z} \sqrt n}\right).
$$
We also want~\eqref{cond} to hold.
These four conditions hold for
$$
n=n(z)= \left\lfloor \frac{\ln |1-4z|}{\ln |z|}\right\rfloor.
$$
Indeed, for this choice of $n$, we have
$$
|z|^n = O(1-4z)  \quad \hbox{ and } \quad 
 (4z)^n =1+O\big(|1-4z| \ln |1-4z|\big),
$$
so that 
\begin{eqnarray*}
\mathbf \KN ^{(1)}(n,z)&=& O\left( (1-4z)^{-1/2} \ln^{-3/2} \frac 1{|1-4z|}\right),
\\
\mathbf \KN ^{(3)}(n,z)&=& O\left( (1-4z)^{-1/2} \ln^{-3/2} \frac 1{|1-4z|}\right),
\\
\mathbf \KN ^{(2)}(n,z)&=& \frac{8\sqrt {\ln
4}}{ \sqrt \pi \sqrt{1-4z} \sqrt{\ln\frac1{|1-4z|}}}\left(1+  O\left(\ln^{-1} \frac 1{|1-4z|}\right)\right).
\end{eqnarray*}
Finally, since 
$$
\ln |1-4z| = \ln(1-4z) \left(1+ O \left(\ln^{-1} \frac 1{|1-4z|}\right)\right),
$$
 and because $z$ is close to $\frac{1}{4}$, we have at last obtained
$$
\mathbf \KN (z) 
 =\frac{8\sqrt {\ln
4}}{\sqrt \pi \sqrt{1-4z} \sqrt{\ln\frac1{1-4z}}} \left(1+ O \left(\ln^{-1} \frac 1{1-4z}\right)\right),
$$
as stated in 
Proposition~\ref{prop:m=1}.
\qed


\subsection{Proof of Proposition~\ref{lemme1-fl} (case $m\ge 2$)}
\label{sec:binary-m}
      The proof is a straightforward adaptation of the proof of
      Theorem~\ref{theorem:asymptotic_binary_node}. We simply describe here the few necessary changes.
      We start from the expression of Proposition~\ref{lemma:binary_K}:
      \beq\label{K-new}
      \mathbf  \KN _m(z) := \mathbf    \KN _m^{\mathcal B} (z) =\frac 1 {2z^2}  \sum_{p\ge 0} B_{p} m^{p} \left( \sqrt{1-4mz+4mz^{p+2}} - \sqrt{1-4mz}\right)
      \eeq
      where  $B_p$ still denotes the 
$(p+1)^{\textrm{st}}$-Catalan number.

Let us first prove that the series $\mathbf \KN _m(z)$ is  analytic in the domain $D$ defined
      by $|z|<1/(2m)$ and $z \not \in [1/(4m), 1/(2m)]$.
The counterpart of Lemma~\ref{lem:rouche} states that 
       $1-4mz+4mz^{p+2}$ has exactly one root of modulus
less than $1/(2m)$, and that this root is larger than $1/(4m)$. This
now holds for any $p\ge 0$ (provided $m\ge 2$).
Hence each series $\mathbf C_{m,p}^{\mathcal B}(z)$ is thus analytic
in $D$.
We also prove that $u_p:=4mz^{p+2}/(1-4mz)$ avoids the half-line
$(-\infty, -1]$ for $z\in D$.

The proof that $\mathbf \KN _m(z)$ is also analytic in $D$ transfers verbatim, once we
have written
\begin{eqnarray*}
 \mathbf   \KN _m(z)&=&
%
\frac{m}{\sqrt{1-4mz}}\mathbf B(mz)
%
   + \frac{\sqrt{1-4mz}}{2z^2} \sum_{p\ge 0}B_{p}
   m^{p}\left(\sqrt{1+u_p}-1-\frac{u_p}2\right).
\end{eqnarray*}
In Step 2, we split $\mathbf \KN _m(z)$ in the three following parts:
\begin{eqnarray}
\mathbf \KN _m^{(1)}(n,z)&=& \frac{\sqrt{1-4mz}}{2z^2} \sum_{p=0}^n
B_{p} m^{p}\left( \sqrt{1+u_p}-1\right),\nonumber
\\
\mathbf \KN _m^{(2)}(n,z)&=&
\frac m{\sqrt{1-4mz}} \left( \mathbf B(mz) - \sum_{p=0}^n B_pm^{p}
z^p\right),\nonumber 
\\
\mathbf \KN _m^{(3)}(n,z)&=& \frac{\sqrt{1-4mz}}{2z^2} \sum_{p>n}B_{p}m^{p} \left(\sqrt{1+u_p}-1-\frac{u_p}2\right).\nonumber 
\end{eqnarray}
Now combining Lemmas~\ref{lem:b1} and~\ref{lem:sum-div} gives an upper bound for
$\mathbf \KN_m ^{(1)}$:
$$
\mathbf \KN _m^{(1)}(n,z)= O\left((4m)^n |z|^{n/2} n^{-3/2}\right),
$$
uniformly in $n$ and $z$ taken in some neighborhood of $1/(4m)$.
The upper bound on $\mathbf \KN_m ^{(3)}$ is found to be
$$
\mathbf \KN_m ^{(3)}(n,z) =O\left(  \frac{m^{n} |z|^{2n}}{|1-4mz| ^{3/2}}
B_n\right).
$$
Finally, since  $\mathbf \KN_m ^{(2)}(n,z)$ is simply 
$ m \mathbf \KN ^{(2)}(n,mz)$ 
(with      $\mathbf \KN ^{(2)}$ defined by~\eqref{K2}), we derive
from~\eqref{K2-est} that
\beq\label{K2-est-m}
\mathbf \KN_m ^{(2)}(n,z)=\frac {8 (4mz)^{n+1} } {\sqrt{\pi} \sqrt{1-4mz}\sqrt n} \left(1 + O(\sqrt{n(1-4mz)})
 +O(n^{-1})\right).
\eeq
The threshold function is now 
$$
n=n(z)= \left\lfloor \frac{\ln |1-4mz|}{\ln |z|}\right\rfloor,
$$
and the rest of the proof follows verbatim, leading to
$$
\mathbf \KN _m(z) 
 =\frac{8\sqrt {\ln
(4m)}}{\sqrt \pi \sqrt{1-4mz} \sqrt{\ln\frac1{1-4mz}}} \left(1+ O \left(\ln^{-1} \frac 1{1-4mz}\right)\right),
$$
which concludes the proof of Proposition~\ref{lemme1-fl}.\qed
%
\subsection{Proof of Theorem~\ref{lemma:asymptotic_binary_edge}}
\label{sec:binary-e}
The series $\mathbf E^{\mathcal B}_m(z)$ given in Proposition~\ref{lemma:binary_K} can be written as in~\eqref{K-new}, upon replacing the numbers $B_p$ by
$$
\bar B_p= \frac{3p}{2p+1} B_p,
$$
with \gf\
$$
\sum_{p \ge 0} \bar B_p z^p= \frac{1-3z -(1-z) \sqrt{1-4z}}{z^2}.
$$
One can then adapt the proof of Proposition~\ref{lemme1-fl} without any
difficulty. The only significant change  is in the
estimate~\eqref{K2-est} (and more generally~\eqref{K2-est-m}) of
$\mathbf \KN_m^{(2)}(n,z)$, which is multiplied by a factor $3/2$. This leads to
the factor $3$ in Theorem~\ref{lemma:asymptotic_binary_edge}.

\subsection{Proof of
  Theorem~\ref{theorem:asymtotic_node_size_unranked}}
\label{sec:unranked-n}

The proof is again a variation on the proof of
Theorem~\ref{theorem:asymptotic_binary_node}. Let us describe it
directly for a general value of $m$.

Our first objective is to obtain the following counterpart of
Proposition~\ref{lemme1-fl}: The generating function $\mathbf \KN _m^{\mathcal{T}}(z)$
 is analytic in the domain $D$ 
defined by $|z|< \frac{1}{2m}$ and $z \notin
[ \frac{1}{4m},\frac{1}{2m} ]$. As $z$ tends to $ \frac{1}{4m}$ in $D$,
one has
\begin{equation}\label{N-unranked}
\mathbf  \KN_m ^{\mathcal{T}}(z) = 
\frac{m \kappa_m } {\sqrt{(1-4m z ) \ln ((1-4mz)^{-1})}} + O \left(\frac{1}{\sqrt{(1-4mz)\ln^3((1-4mz)^{-1})}} \right),
\end{equation}
where $\kappa_m$ is defined as in Theorem~\ref{theorem:asymptotic_binary_node}.

The transfer theorem from~\cite{AnalyticCombinatorics} then gives 
$$
N^{\mathcal{T}}_{m,n} = [z^n]\mathbf \KN_m ^{\mathcal{T}}(z)
= \frac{\kappa_m}{\sqrt{\pi}}
\frac {4^{n} m^{n+1}}{\sqrt{n \ln n}}\left(1+ O(\ln^{-1} n)\right).
$$
Since the  number of $m$-labeled unranked trees of size $n$ is
$$
T_{m,n} \sim \frac{4^{n} m^{n+1}}{\sqrt \pi n^{3/2}}\left(1+ O(n^{-1})\right),
$$
this gives for the average number of nodes the expression of
Theorem~\ref{theorem:asymtotic_node_size_unranked}.

So let us focus on the proof of~\eqref{N-unranked}, which will mimic
the proof of Proposition~\ref{lemme1-fl}.
We start from the expression of Proposition~\ref{prop:NE-unranked}:
$$
\mathbf N_m(z):= \mathbf \KN_m^{\mathcal{T}}(z)= \sum_{p\ge 0} T_p m^{p+1}\mathbf C_{m,p}^{\mathcal{T}}(z)
$$
where
$$
\mathbf C_{m,p}^{\mathcal{T}}(z) = \frac{z^{p+1} + \sqrt{1-4mz+2z^{p+1} + z^{2p+2} }-\sqrt{1-4mz}}{2z} 
$$
and $T_p = {T}_{1,p}$ is the $p$-th Catalan number, with  \gf\
\beq\label{T-B}
\mathbf T(z)=\sum_{p\ge 0} T_p z^p=\frac{1-\sqrt{1-4z}}{2z}
=1+ z \mathbf B(z).
\eeq
Recall that the series $\mathbf B(z)$ is defined by~\eqref{B-def}. We follow the same steps as in the proof of Proposition~\ref{lemme1-fl}.

\paragraph{Step 1: $\mathbf N_m(z)$ is analytic in $D$.}
With Rouch\'e's theorem we can prove that
$\mathbf C_{m,p}^{\mathcal{T}}(z)$
is analytic in the domain $D$.
We then define $u_p= z^{p+1}(2+z^{p+1})/(1-4mz)$ and
prove that $u_p$ does not meet the half-line $(-\infty,-1]$ for $z
  \in D$. We then write
$$
\mathbf N_m(z)= \frac m 2 \mathbf T(mz)+ \frac {m\mathbf T(mz)}{2\sqrt{1-4z}}  +
\frac{mz\mathbf T(mz^2)}{4\sqrt{1-4mz}}  +
\frac{\sqrt{1-4mz}}{2z} \sum_{p\ge 0} T_p m^{p+1} \left(
\sqrt{1+u_p}-1 - \frac{u_p}2\right)
$$
to conclude, with the arguments of Sections~\ref{proof:fl} and~\ref{sec:binary-m}, that $\mathbf N_m(z)$ is also analytic in $D$.

\paragraph{Step 2: splitting $\mathbf N_m(z)$.}
We fix an integer $n$ and write
\begin{equation*}\label{K-split_u}
\mathbf N_m(z)=\mathbf N_m^{(1)}(n,z)+\mathbf N^{(2)}_m(n,z)+\mathbf N^{(3)}_m(n,z)+\mathbf N^{(4)}_m(n,z),
\end{equation*}
where
\begin{eqnarray*}
\mathbf N_m^{(1)}(n,z)&=& \frac{\sqrt{1-4mz}}{2z} \sum_{p=0}^n
T_p m^{p+1}\left( \sqrt{1+u_p}-1\right),\label{K1_u}
\\
\mathbf N_m^{(2)}(n,z)&=&   \frac m{2\sqrt{1-4mz}} \left(
\mathbf T(mz)- \sum_{p=0}^n T_p m ^pz^p\right),\label{K2_u}
\\
\mathbf N _m^{(3)}(n,z)&=&  \frac{\sqrt{1-4mz}}{2z} \sum_{p>n}T_p m^{p+1}\left(\sqrt{1+u_p}-1-\frac{u_p}2\right),\label{K3_u}
\\
\mathbf N_m^{(4)}(n,z)&=& 
\frac{m}{2} \mathbf T(mz) +
 \frac{zm}{4 \sqrt{1-4mz}} \left( \mathbf T(mz^2)- \sum_{p=0}^n T_p m^pz^{2p}\right).\label{K4_u}
\end{eqnarray*}
One readily checks that~\eqref{K-split_u} indeed holds. Moreover, each
$\mathbf N^{(i)}_m(n,z)$ is analytic in $D$ for any $i$ and $n$.

\paragraph{Step 3: an upper bound on $\mathbf N^{(1)}_m$.}
Using the same ingredients as before, we find that, as in the binary case,
$$
 \mathbf N_m^{(1)}(n,z) = O\left( (4m)^n|z|^{n/2}n^{-3/2}\right)
$$
uniformly in $n$ and $z$ taken in a neighborhood of $1/(4m)$.
\paragraph{Step 4: an upper bound on $\mathbf N^{(3)}_m$.} Again, the behavior
remains the same as in the binary case:
$$
\mathbf N^{(3)}_m= O\left(\frac{m^n |z|^{2n}}{|1-4mz|^{3/2}} T_n \right),
$$
uniformly in $n$ and $z$, for $z$ in $D$ close enough to $1/(4m)$.
\paragraph{Step 5: estimate of $\mathbf N^{(2)}_m$.}
Using~\eqref{T-B}, we observe that  
$$
\mathbf N_m^{(2)}(n,z) = \frac{m^2z}2 \mathbf N^{(2)}(n-1,mz),$$
 with $\mathbf N^{(2)}(z)$ defined by~\eqref{K2}. Hence the
  estimate~\eqref{K2-est} gives
 \beq\label{N2-est-m}
\mathbf N_m^{(2)}(n,z) = \frac {m (4mz)^{n+2}}{\sqrt \pi \sqrt{1-4mz} \sqrt n} 
 \left(1 + O(\sqrt{n(1-4mz)})
 +O(n^{-1})\right).
\eeq
\paragraph{Step 6: an upper bound of $\mathbf N^{(4)}_m$.}
Given that $m|z|^2 <1/(4m)\le 1/4$, we can write
\begin{eqnarray*}
\left| \mathbf N^{(4)}_m(n,z) \right|&\le &
\frac m 2 |\mathbf T(mz)| + \frac{m|z|}{4 \sqrt{|1-4mz|}} \sum_{p>n} T_p
m^{p}|z|^{2p}
\\
          &=& O \left(1+ \frac{T_n m^n z^{2n} }{\sqrt{|1-4mz|} } \right)
\end{eqnarray*}
for $z$ in a neighborhood of $1/(4m)$.
The argument is the same as for the bound~\eqref{K3-bound}.

\paragraph{Step 7: the threshold $n(z)$.} 
Using the same threshold function as in the binary case, we see that $\mathbf N_m(n,z)$
is dominated by $\mathbf N^{(2)}_m(n,z)$,
and more precisely, that~\eqref{N-unranked} holds.

\subsection{Proof of Theorem~\ref{thm:unranked-e}}
\label{sec:unranked-e}
Comparing the two series of Proposition~\ref{prop:NE-unranked} shows that it suffices
to replace the numbers $T_p$ by
$$
\bar T_p = \frac {3p}{p+2} T_p,
$$
with \gf\
$$
\sum_{p \ge 0} \bar T_p z^p= \frac{1-3z -(1-z) \sqrt{1-4z}}{2z^2},
$$
to go from $\mathbf N^{\mathcal T}_m(z)$ to $\mathbf E^{\mathcal T}_m(z)$.
One can then adapt the proof of Proposition~\ref{theorem:asymtotic_node_size_unranked} without any
difficulty. The only significant change  is in the
estimate~\eqref{N2-est-m} of
$\mathbf \KN_m^{(2)}(n,z)$, which is multiplied by a factor $3$. This leads to
the factor $3$ in Theorem~\ref{thm:unranked-e}.

\end{document}

%% file: tn.pdf_t
\begin{picture}(0,0)%
\includegraphics{tn.pdf}%
\end{picture}%
\setlength{\unitlength}{3868sp}%
\begingroup\makeatletter\ifx\SetFigFont\undefined%
\gdef\SetFigFont#1#2#3#4#5{%
  \reset@font\fontsize{#1}{#2pt}%
  \fontfamily{#3}\fontseries{#4}\fontshape{#5}%
  \selectfont}%
\fi\endgroup%
\begin{picture}(2957,1518)(5137,-5366)
\put(6226,-4486){\makebox(0,0)[b]{\smash{{\SetFigFont{14}{16.8}{\familydefault}{\mddefault}{\updefault}{\color[rgb]{0,0,0}$\dots$}%
}}}}
\put(6001,-5311){\makebox(0,0)[b]{\smash{{\SetFigFont{9}{10.8}{\familydefault}{\mddefault}{\updefault}{\color[rgb]{0,0,0}$n$ times}%
}}}}
\put(7848,-4186){\makebox(0,0)[b]{\smash{{\SetFigFont{6}{7.2}{\familydefault}{\mddefault}{\updefault}{\color[rgb]{0,0,0}$\dots$}%
}}}}
\put(8079,-4235){\makebox(0,0)[lb]{\smash{{\SetFigFont{9}{10.8}{\familydefault}{\mddefault}{\updefault}{\color[rgb]{0,0,0}$n$ edges}%
}}}}
\end{picture}%

%% file: sn.pdf_t
\begin{picture}(0,0)%
\includegraphics{sn.pdf}%
\end{picture}%
\setlength{\unitlength}{3947sp}%
\begingroup\makeatletter\ifx\SetFigFont\undefined%
\gdef\SetFigFont#1#2#3#4#5{%
  \reset@font\fontsize{#1}{#2pt}%
  \fontfamily{#3}\fontseries{#4}\fontshape{#5}%
  \selectfont}%
\fi\endgroup%
\begin{picture}(4237,2727)(1977,-5727)
\put(5854,-4718){\makebox(0,0)[b]{\smash{{\SetFigFont{12}{14.4}{\familydefault}{\mddefault}{\updefault}{\color[rgb]{0,0,0}$\ddots$}%
}}}}
\put(3074,-3642){\makebox(0,0)[b]{\smash{{\SetFigFont{14}{16.8}{\familydefault}{\mddefault}{\updefault}{\color[rgb]{0,0,0}$\dots$}%
}}}}
\put(3061,-4137){\makebox(0,0)[b]{\smash{{\SetFigFont{14}{16.8}{\familydefault}{\mddefault}{\updefault}{\color[rgb]{0,0,0}$\dots$}%
}}}}
\put(3092,-4496){\makebox(0,0)[b]{\smash{{\SetFigFont{9}{10.8}{\familydefault}{\mddefault}{\updefault}{\color[rgb]{0,0,0}$n-1$ times}%
}}}}
\put(3058,-5317){\makebox(0,0)[b]{\smash{{\SetFigFont{14}{16.8}{\familydefault}{\mddefault}{\updefault}{\color[rgb]{0,0,0}$\dots$}%
}}}}
\put(2088,-3749){\rotatebox{90.0}{\makebox(0,0)[rb]{\smash{{\SetFigFont{9}{10.8}{\familydefault}{\mddefault}{\updefault}{\color[rgb]{0,0,0}$n$ times}%
}}}}}
\put(2342,-4542){\makebox(0,0)[b]{\smash{{\SetFigFont{14}{16.8}{\familydefault}{\mddefault}{\updefault}{\color[rgb]{0,0,0}$\vdots$}%
}}}}
\put(3089,-5676){\makebox(0,0)[b]{\smash{{\SetFigFont{9}{10.8}{\familydefault}{\mddefault}{\updefault}{\color[rgb]{0,0,0}$n-1$ times}%
}}}}
\put(4377,-3954){\rotatebox{90.0}{\makebox(0,0)[rb]{\smash{{\SetFigFont{9}{10.8}{\familydefault}{\mddefault}{\updefault}{\color[rgb]{0,0,0}$n$ edges}%
}}}}}
\put(4653,-4514){\makebox(0,0)[b]{\smash{{\SetFigFont{14}{16.8}{\familydefault}{\mddefault}{\updefault}{\color[rgb]{0,0,0}$\vdots$}%
}}}}
\put(4996,-5648){\makebox(0,0)[b]{\smash{{\SetFigFont{9}{10.8}{\familydefault}{\mddefault}{\updefault}{\color[rgb]{0,0,0}$n$ edges}%
}}}}
\put(5551,-5236){\makebox(0,0)[b]{\smash{{\SetFigFont{9}{10.8}{\familydefault}{\mddefault}{\updefault}{\color[rgb]{0,0,0}$n-2$ edges}%
}}}}
\end{picture}%

%% file: hybrid.pdf_t
\begin{picture}(0,0)%
\includegraphics{hybrid.pdf}%
\end{picture}%
\setlength{\unitlength}{3552sp}%
\begingroup\makeatletter\ifx\SetFigFont\undefined%
\gdef\SetFigFont#1#2#3#4#5{%
  \reset@font\fontsize{#1}{#2pt}%
  \fontfamily{#3}\fontseries{#4}\fontshape{#5}%
  \selectfont}%
\fi\endgroup%
\begin{picture}(6206,4441)(5093,-7955)
\put(7748,-6715){\makebox(0,0)[lb]{\smash{{\SetFigFont{9}{10.8}{\familydefault}{\mddefault}{\updefault}{\color[rgb]{0,0,0}dag}%
}}}}
\put(9304,-7891){\makebox(0,0)[lb]{\smash{{\SetFigFont{9}{10.8}{\familydefault}{\mddefault}{\updefault}{\color[rgb]{0,0,0}$|$hdag$(t)| = 5$}%
}}}}
\put(6129,-6644){\makebox(0,0)[lb]{\smash{{\SetFigFont{9}{10.8}{\familydefault}{\mddefault}{\updefault}{\color[rgb]{0,0,0}A}%
}}}}
\put(7487,-6302){\makebox(0,0)[lb]{\smash{{\SetFigFont{9}{10.8}{\familydefault}{\mddefault}{\updefault}{\color[rgb]{0,0,0}S}%
}}}}
\put(9504,-6294){\makebox(0,0)[lb]{\smash{{\SetFigFont{9}{10.8}{\familydefault}{\mddefault}{\updefault}{\color[rgb]{0,0,0}S}%
}}}}
\put(10788,-6694){\makebox(0,0)[lb]{\smash{{\SetFigFont{9}{10.8}{\familydefault}{\mddefault}{\updefault}{\color[rgb]{0,0,0}A}%
}}}}
\put(7300,-5752){\makebox(0,0)[lb]{\smash{{\SetFigFont{9}{10.8}{\familydefault}{\mddefault}{\updefault}{\color[rgb]{0,0,0}fcns}%
}}}}
\put(7351,-3661){\makebox(0,0)[lb]{\smash{{\SetFigFont{9}{10.8}{\familydefault}{\mddefault}{\updefault}{\color[rgb]{0,0,0}$|$dag$(t)| = 6 = |$bdag$(t)|$}%
}}}}
\end{picture}%

%% file: tocs.bbl
\begin{thebibliography}{10}

\bibitem{DBLP:journals/toit/ArionBMP07}
A.~Arion, A.~Bonifati, I.~Manolescu, and A.~Pugliese.
\newblock {XQueC}: A query-conscious compressed {XML} database.
\newblock {\em ACM Trans. Internet Techn.}, 7(2), 2007.

\bibitem{DBLP:journals/pvldb/BakibayevOZ12}
N.~Bakibayev, D.~Olteanu, and J.~Zavodny.
\newblock Fdb: A query engine for factorised relational databases.
\newblock {\em PVLDB}, 5(11):1232--1243, 2012.

\bibitem{DBLP:conf/soda/BilleLRSSW11}
P.~Bille, G.~M. Landau, R.~Raman, K.~Sadakane, S.~R. Satti, and O.~Weimann.
\newblock Random access to grammar-compressed strings.
\newblock In {\em SODA}, pages 373--389, 2011.

\bibitem{DBLP:conf/vldb/KochBG03}
P.~Buneman, M.~Grohe, and C.~Koch.
\newblock Path queries on compressed {XML}.
\newblock In {\em VLDB}, pages 141--152, 2003.

\bibitem{DBLP:journals/is/BusattoLM08}
G.~Busatto, M.~Lohrey, and S.~Maneth.
\newblock Efficient memory representation of {XML} document trees.
\newblock {\em Inf. Syst.}, 33(4-5):456--474, 2008.

\bibitem{debruijn}
N.~G. de~Bruijn, D.~E. Knuth, and S.~O. Rice.
\newblock The average height of planted plane trees.
\newblock In {\em Graph theory and computing}, pages 15--22. Academic Press,
  New York, 1972.

\bibitem{DBLP:journals/dm/DershowitzZ80}
N.~Dershowitz and S.~Zaks.
\newblock Enumerations of ordered trees.
\newblock {\em Discrete Mathematics}, 31(1):9--28, 1980.

\bibitem{DBLP:journals/jacm/DowneyST80}
P.~J. Downey, R.~Sethi, and R.~E. Tarjan.
\newblock Variations on the common subexpression problem.
\newblock {\em J. ACM}, 27(4):758--771, 1980.

\bibitem{DBLP:journals/cacm/Ershov58}
A.~P. Ershov.
\newblock On programming of arithmetic operations.
\newblock {\em Commun. ACM}, 1(8):3--9, 1958.

\bibitem{flajolet-trees}
P.~Flajolet and A.~Odlyzko.
\newblock The average height of binary trees and other simple trees.
\newblock {\em J. Comput. System Sci.}, 25(2):171--213, 1982.

\bibitem{Fl-Od}
P.~Flajolet and A.~Odlyzko.
\newblock Singularity analysis of generating functions.
\newblock {\em SIAM J. Discrete Math.}, 3(2):216--240, 1990.

\bibitem{AnalyticCombinatorics}
P.~Flajolet and R.~Sedgewick.
\newblock {\em Analytic Combinatorics}.
\newblock Cambridge University Press, 2009.

\bibitem{FlaSipStey1990}
P.~Flajolet, P.~Sipala, and J.-M. Steyaert.
\newblock Analytic variations on the common subexpression problem.
\newblock In {\em ICALP}, pages 220--234, 1990.

\bibitem{DBLP:books/aw/Knuth68}
D.~E. Knuth.
\newblock {\em The Art of Computer Programming, Vol. I: Fundamental
  Algorithms}.
\newblock Addison-Wesley, 1968.

\bibitem{DBLP:conf/vldb/Koch03}
C.~Koch.
\newblock Efficient processing of expressive node-selecting queries on {XML}
  data in secondary storage: A tree automata-based approach.
\newblock In {\em VLDB}, pages 249--260, 2003.

\bibitem{DBLP:conf/dcc/LarssonM99}
N.~J. Larsson and A.~Moffat.
\newblock Offline dictionary-based compression.
\newblock In {\em DCC}, pages 296--305, 1999.

\bibitem{DBLP:conf/sigmod/LiefkeS00}
H.~Liefke and D.~Suciu.
\newblock {XMILL}: An efficient compressor for {XML} data.
\newblock In {\em SIGMOD Conference}, pages 153--164, 2000.

\bibitem{l12}
M.~Lohrey.
\newblock Algorithmics on {SLP}-compressed strings: a survey.
\newblock {\em Groups Complexity Cryptology}, 4:241--299, 2013.

\bibitem{DBLP:journals/tcs/LohreyM06}
M.~Lohrey and S.~Maneth.
\newblock The complexity of tree automata and {XP}ath on grammar-compressed
  trees.
\newblock {\em Theor. Comput. Sci.}, 363(2):196--210, 2006.

\bibitem{lohmanmen13}
M.~Lohrey, S.~Maneth, and R.~Mennicke.
\newblock {XML} tree structure compression using repair.
\newblock {\em Inf. Syst.}, 38(8):1150--1167, 2013.

\bibitem{DBLP:conf/icdt/LohreyMN13}
M.~Lohrey, S.~Maneth, and E.~Noeth.
\newblock {XML} compression via dags.
\newblock In {\em ICDT}, pages 69--80, 2013.

\bibitem{DBLP:journals/jcss/LohreyMS12}
M.~Lohrey, S.~Maneth, and M.~Schmidt-Schau{\ss}.
\newblock Parameter reduction and automata evaluation for grammar-compressed
  trees.
\newblock {\em J. Comput. Syst. Sci.}, 78(5):1651--1669, 2012.

\bibitem{DBLP:journals/corr/abs-1012-5696}
S.~Maneth and T.~Sebastian.
\newblock Fast and tiny structural self-indexes for {XML}.
\newblock {\em CoRR}, abs/1012.5696, 2010.

\bibitem{marckert}
J.-F. Marckert.
\newblock The rotation correspondence is asymptotically a dilatation.
\newblock {\em Random Structures Algorithms}, 24(2):118--132, 2004.

\bibitem{DBLP:books/sp/MeinelT98}
C.~Meinel and T.~Theobald.
\newblock {\em Algorithms and Data Structures in VLSI Design: OBDD -
  Foundations and Applications}.
\newblock Springer, 1998.

\bibitem{DBLP:journals/sigmod/Neven02}
F.~Neven.
\newblock Automata theory for {XML} researchers.
\newblock {\em SIGMOD Record}, 31(3):39--46, 2002.

\bibitem{DBLP:journals/jair/Nevill-ManningW97}
C.~G. Nevill-Manning and I.~H. Witten.
\newblock Identifying hierarchical strcture in sequences: A linear-time
  algorithm.
\newblock {\em J. Artif. Intell. Res. (JAIR)}, 7:67--82, 1997.

\bibitem{DBLP:conf/esa/Plandowski94}
W.~Plandowski.
\newblock Testing equivalence of morphisms on context-free languages.
\newblock In {\em ESA}, pages 460--470, 1994.

\bibitem{DBLP:journals/jcss/Schwentick07}
T.~Schwentick.
\newblock Automata for {XML} - a survey.
\newblock {\em J. Comput. Syst. Sci.}, 73(3):289--315, 2007.

\bibitem{DBLP:conf/dbpl/Suciu01}
D.~Suciu.
\newblock Typechecking for semistructured data.
\newblock In {\em DBPL}, pages 1--20, 2001.

\end{thebibliography}
